\documentclass[11pt]{article}
\usepackage{amsfonts,amsmath,amsthm,amssymb}

\usepackage{comment, enumerate}
\usepackage[margin=1in]{geometry}
\usepackage{soul,color}
\usepackage{mathrsfs}
\usepackage{mathdots}
\usepackage{algorithm}
\usepackage{algpseudocode}
\usepackage{caption}
\usepackage{subcaption}
\usepackage[pagebackref]{hyperref}
\hypersetup{colorlinks=true,citecolor=red,linkcolor=blue}
\usepackage{multicol}
\usepackage[small,compact]{titlesec}
\usepackage{times}
\usepackage{paralist}

\usepackage{bbm}

\usepackage{url}

\usepackage{tabu}

\usepackage{xfrac}

\usepackage[normalem]{ulem}

\newcommand{\X}{\mathsf{X}}
\newcommand{\Y}{\mathsf{Y}}
\newcommand{\Z}{\mathsf{Z}}

\newcommand{\poly}{\mathop{\mathrm{poly}}}

\newcommand{\rank}{\mathop{\operatorname{rank}}}

\newcommand{\var}{\mathop{\operatorname{var}}}

\newcommand{\eff}{\mathop{\operatorname{eff}}}
\newcommand{\effQ}{\mathop{\operatorname{effQ}}}
\newcommand{\Peff}{P\text{-}\eff}
\newcommand{\wt}{\widetilde}

\newcommand{\E}{\mathbb{E}}
\newcommand{\R}{\mathbb{R}}
\newcommand{\N}{\mathbb{N}}

\newcommand{\avg}{\mathrm{avg}}

\newcommand{\eps}{\varepsilon}

\theoremstyle{plain}\newtheorem{theorem}{Theorem}[section]
\newtheorem{lemma}[theorem]{Lemma}
\newtheorem{claim}[theorem]{Claim}

\newtheorem{example}[theorem]{Example}

\newtheorem{fact}[theorem]{Fact}
\newtheorem{remark}[theorem]{Remark}
\theoremstyle{definition}
\newtheorem{definition}[theorem]{Definition}

\definecolor{zhj}{RGB}{255,50,200}
\definecolor{josh}{RGB}{150,50,50}

\usepackage[utf8]{inputenc}

\title{Generalizations of Matrix Multiplication\\can solve the Light Bulb Problem}

\author{
Josh Alman\thanks{\texttt{josh@cs.columbia.edu}. Columbia University. Supported in part by NSF Grant CCF-2238221 and a grant from the Simons Foundation
(Grant Number 825870 JA).}
\quad
Hengjie Zhang\thanks{\texttt{hengjie.z@columbia.edu}. Columbia University. Supported in part by NSF grant CCF2008733 and ONR grant N00014-22-1-2713.}
}

\begin{document}

\maketitle

\begin{abstract}
In the light bulb problem, one is given as input vectors $x_1, \ldots, x_n, y_1, \ldots, y_n \in \{-1,1\}^d$ which are all uniformly random. They are all chosen independently except for a planted pair $(x_{i^*}, y_{j^*})$ which is chosen to have correlation $\rho$ for some constant $\rho>0$. The goal is to find the planted pair. The light bulb problem was introduced over 30 years ago by L.~Valiant, and is known to have many applications in data analysis, statistics, and learning theory.

The naive algorithm runs in $\Omega(n^2)$ time, and algorithms based on Locality-Sensitive Hashing approach quadratic time as $\rho \to 0$. In 2012, G.~Valiant gave a breakthrough algorithm running in time $O(n^{(5-\omega)/(4-\omega)}) < O(n^{1.615})$, no matter how small $\rho>0$ is, by making use of fast matrix multiplication. This was subsequently refined by Karppa, Kaski, and Kohonen in 2016 to running time $O(n^{2 \omega / 3}) < O(n^{1.582})$, but is essentially the only known approach for this important problem.

In this paper, we propose a new approach based on replacing fast matrix multiplication with other \emph{variants} and \emph{generalizations} of matrix multiplication, which can be computed faster than matrix multiplication, but which may omit some terms one is supposed to compute, and include additional error terms. Our new approach can make use of a wide class of tensors which previously had no known algorithmic applications, including tensors which arise naturally as intermediate steps in border rank methods and in the Laser method.

We further show that our approach can be combined with locality-sensitive hashing to design an algorithm whose running time improves as $\rho$ gets larger. To our knowledge, this is the first algorithm which combines fast matrix multiplication with hashing for the light bulb problem or any closest pair problem, and it leads to faster algorithms for small $\rho>0$.

We then focus on tensors for ``multiplying'' $2 \times 2$ matrices; using such small tensors is typically required for practical algorithms. 
In this setting, the best prior algorithm, using Strassen's algorithm for matrix multiplication, yields a running time of only $O(n^{1.872})$.
We introduce a new such low-rank tensor we call $T_{2112}$, which has omissions and errors compared to matrix multiplication, and using it, we design a new algorithm for the light bulb problem which runs in time $O(n^{1.797})$. We also explain why we are optimistic that this approach could yield asymptotically faster algorithms for the light bulb problem.
\end{abstract}

\thispagestyle{empty}
\newpage
\setcounter{page}{1}

\section{Introduction} \label{sec:intro}

We've known since the work of Strassen~\cite{strassen1973vermeidung} that designing algebraic algorithms for matrix multiplication is equivalent to bounding the ranks of matrix multiplication tensors. Since then, an enormous amount of work has gone into bounding the ranks of these tensors in various regimes, combining techniques from algebra, combinatorics, algorithm design, and computer search. One big reason that so much effort has gone into this problem is that matrix multiplication has many algorithmic applications; algorithmic problems from nearly every domain of computation have been reduced to matrix multiplication.

A goal of this paper is to show that tensors other than matrix multiplication can also have algorithmic applications. In this way, the same techniques which have been developed for matrix multiplication could be repurposed to lead to new algorithms. Tensors whose support is a subset of the support of matrix multiplication have been applied to Boolean matrix multiplication~\cite{cohn2013fast,karppa2019probabilistic,harris2021improved} and even directly to matrix multiplication~\cite{schonhage1981partial}, but we're unaware of applications of any tensors whose support is incomparable with matrix multiplication (other than a small handful of special problems like polynomial multiplication).

We focus here on the light bulb problem, a fundamental problem from learning theory which currently has two best algorithms depending on the parameter regime: one using fast matrix multiplication, and one using locality-sensitive hashing. Somewhat surprisingly, there are no known parameter regimes where combining the two approaches leads to an improved algorithm. Using tensors other than matrix multiplication, we achieve two main results
\begin{enumerate}
    \item Any tensor can replace matrix multiplication to solve the light bulb problem. The efficiency of this algorithm comes from a trade-off between the tensor's rank and how similar it is to matrix multiplication. We find that, restricted to small tensors, there are better tensors than matrix multiplication that lead to improved algorithms.
    \item Tensors other than matrix multiplication can be combined with locality-sensitive hashing to yield improved algorithms. We find that the symmetry of matrix multiplication prevents one from combining it with hashing, but that sufficiently asymmetric tensors can be improved with hashing.
\end{enumerate}

\subsection{The Light Bulb Problem}

In the light bulb problem for $n$ vectors of dimension $d$ and correlation $\rho > 0$, we are given as input vectors $x_1, \ldots, x_n, y_1, \ldots, y_n \in \{-1,1\}^d$ which are all picked uniformly at random, and all picked independently except for a `planted pair' $(x_{i^*}, y_{j^*})$ which is chosen to have correlation $\geq \rho$ (i.e., so that $\langle x_{i^*}, y_{j^*} \rangle \geq \rho \cdot d $). The indices $i^*$ and $j^*$ of the planted pair are unknown to us, and our goal is to find them.\footnote{The light bulb problem is often stated as a `monochromatic' problem, where we are not told which are `$x$' or `$y$' vectors, but this has a simple reduction to the `bichromatic' version we define here.}

L.~Valiant introduced the light bulb problem over 30 years ago~\cite{lightbulb} as a basic primitive which captures the fundamental task of detecting correlations among $n$ random variables. It can be seen as a special case of a multitude of other problems in data analysis, statistics, and learning theory, and for many of these problems, the fastest known algorithm comes from a reduction to the light bulb problem. For instance:
\begin{compactitem}
    \item If one would like to detect correlations among random variables with a range $R$ other than just $\{-1,1\}$, one can typically reduce to the light bulb problem by making use of a locality-sensitive hash function for $R$. For example, if $R$ is the Euclidean sphere, then one can map points in $R$ to $\{-1,1\}$ by determining which side of a random hyperplane they lie on, which only slightly decreases the correlation $\rho$, by a constant factor~\cite{rounding}.
    \item The light bulb problem is a special case of many learning problems, including learning sparse parities with noise, and learning $k$-Juntas with and without noise, and the fastest known algorithms for these problems come from reducing the general case to the light bulb problem~\cite{valiant2012finding}.
\end{compactitem}

Suppose $d = \Theta(\log n)$. The straightforward algorithm for this problem simply compares each pair of vectors and runs in time $\tilde{O}(n^2)$.\footnote{We write $\tilde{O}(t)$ to suppress polylog$(t)$ factors.} Techniques for nearest neighbor search like locality-sensitive hashing have been applied to the problem, culminating in Dubiner's algorithm~\cite{dubiner2010bucketing} which runs in time $n^{2/(\rho+1) + o(1)}$. This is the fastest known algorithm for larger $\rho$, but its running time becomes quadratic as $\rho \to 0$. In 2012, G.~Valiant~\cite{valiant2012finding} gave a breakthrough algorithm running in time $O(n^{1.615})$ no matter how small the constant $\rho>0$ is. Thereafter, Karppa, Kaski, and Kohonen~\cite{karppa2018faster} improved the running time to $O(n^{1.582})$. This is faster than Dubiner's algorithm for all $0 < \rho < 0.264$. To emphasize, these algorithms work for any constant $\rho>0$, but give essentially the same running time no matter how large $\rho$ is.

The key ideas behind these latter two algorithms focus on the dimension $d$, which is sometimes called the `sample complexity'. One would typically like to keep $d$ low while still solving the problem quickly. It is information-theoretically necessary to pick $d = \Omega(\log n)$ (since if $d = o(\log n)$, then by the pigeonhole principle, two of the uncorrelated vectors will be \emph{equal to} each other and indistinguishable from the correlated pair).

Interestingly, G.~Valiant~\cite{valiant2012finding} introduced a `XOR/Tensor Embedding' technique, and Karppa, Kaski, and Kohonen~\cite{karppa2018faster} gave a more efficient `compressed matrices' implementation, which (roughly) allows one to efficiently `expand' lower-dimensional vectors, and convert $d = \Theta(\log n)$ to a much larger $d = n^{\Theta(1)}$ with only a negligible decrease in $\rho$. This allows one to focus on the task of designing faster algorithms for detecting correlations without worrying about $d$. More precisely, these prior algorithms consist of two phases solving two different problems: a `vector aggregation' problem of converting groups of shorter vectors into single longer vectors, and a `light bulb computation' problem of actually detecting the correlations among these vectors. The final running time of \cite{valiant2012finding} trades off between the running times of these two problems, and the later work \cite{karppa2018faster} showed how to make the running time of vector aggregation negligible compared to the running time of light bulb computation. Both prior algorithms ultimately solve the light bulb computation problem using fast matrix multiplication, and we focus in this paper on faster algorithms for this problem. (See footnote~\ref{footnote:aggregation} in Section~\ref{sec:firstalg} below for more details.)

Despite the importance of the light bulb problem, no approach beyond locality-sensitive hashing or `expand then use fast matrix multiplication' has been proposed since the breakthrough almost 10 years ago~\cite{valiant2012finding}, and these known approaches seem to have hit their limits~\cite{karppa2018faster,alman2018illuminating}. Furthermore, although hashing approaches and matrix multiplication approaches have been known for both the light bulb problem as well as many other closest pair problems for some time (see, for instance, the survey~\cite{andoni2018approximate}), there are no known algorithms for any of these problems which truly combine the two. 

In this paper, we propose a new approach to designing faster algorithms for the light bulb problem by replacing fast matrix multiplication with other tensors which are \emph{generalizations} of matrix multiplication, and which can be computed faster. We also show how hashing methods can be combined with our approach to design even faster algorithms: while previous matrix multiplication-based algorithms for the light bulb problem have the same running time regardless of how large $\rho>0$ is, our new approach yields algorithms which are faster as $\rho$ gets larger. Before getting into more detail, we introduce some necessary background.

\paragraph{Known Algorithms and Exponents.} 

The exponent of matrix multiplication, $\omega$, is the smallest real number such that for any $\eps > 0$, one can multiply $n \times n$ matrices over a field using $O(n^{\omega + \eps})$ field operations\footnote{The `running time' and `number of field operations' are typically related by low-order terms unless one is working with very large numbers. In principle, $\omega$ might depend on the characteristic of the field, although all known bounds work equally well over any field, so we will abuse notation and simply refer to the same $\omega$ for all fields.}. Since $n \times n$ matrices have $n^2$ entries one must read and write, it is known that $\omega \geq 2$, and the best known algorithms show $\omega < 2.37286$~\cite{CoppersmithW82,stothers,v12,legall,alman2021refined}.

We similarly define the exponent of the light bulb problem, $\omega_\ell$, to be the smallest real number such that for any $\eps>0$, one can solve the light bulb problem with $n$ vectors, for any constant $\rho > 0$, in time $O(n^{\omega_\ell + \eps})$. G.~Valiant~\cite{valiant2012finding} showed that $\omega_\ell \leq (5-\omega)/(4-\omega) < 1.615$, and Karppa, Kaski, and Kohonen~\cite{karppa2018faster} later improved this to the best known bound $\omega_\ell \leq 2 \omega / 3 < 1.582$. Since the input size is only $\tilde{O}(n)$, the corresponding lower bound is $\omega_\ell \geq 1$. However, even showing $\omega = 2$ would only imply that $\omega_\ell \leq 4/3$ using the known algorithms~\cite{valiant2012finding,karppa2018faster,alman2018illuminating}.

We will also discuss the \emph{Boolean} matrix multiplication problem, where we're given as input matrices $A,B \in \{0,1\}^{n \times n}$, and we need to compute the matrix product $C = A \times B$ over the Boolean semiring, i.e., the matrix $C \in \{0,1\}^{n \times n}$ given by $C[i,j] = \bigvee_{k=1}^n (A[i,k] \wedge B[k,j])$. Let $\omega_B$ denote the smallest real number such that for any $\eps>0$, one can solve this problem in time $O(n^{\omega_B + \eps})$. It is known that $2 \leq \omega_B \leq \omega$, and there are no known algorithms for Boolean matrix multiplication that are asymptotically faster than the best known matrix multiplication algorithms (see, e.g.,~\cite[{Section 1}]{karppa2019probabilistic}).

\subsection{Bilinear Problems}

A key technique in this paper will be designing and making use of algorithms for bilinear problems, wherein one would like to evaluate a prescribed set of bilinear polynomials when its variables are set to input numbers. Matrix multiplication is a prominent example, and we will focus particularly on bilinear problems like this where the inputs and outputs are naturally formatted as matrices.
Bilinear problems which take as input a $q_i \times q_k$ matrix $\X$ and a $q_j \times q_k$ matrix $\Y$, and output a $q_i \times q_j$ matrix $\Z$, can be written as a (three-dimensional) tensor $$T = \sum_{i,i' \in [q_i], j,j' \in [q_j], k,k' \in [q_k]} T(\X_{i,k} \Y_{j,k'} \Z_{i',j'}) \cdot \X_{i,k} \Y_{j,k'} \Z_{i',j'},$$
where $T(\X_{i,k} \Y_{j,k'} \Z_{i',j'}) \in \R$ is the coefficient of $\X_{i,k} \Y_{j,k'}$ in the bilinear polynomial we output in entry $\Z_{i',j'}$. One can imagine plugging in values for each of the $\X$ and $\Y$ variables, and then the goal is to compute the coefficient of each $\Z$ variable. For example, for the $q \times q$ matrix multiplication problem $T_q = \langle q,q,q \rangle$, $$T_q(\X_{i,k} \Y_{j,k'} \Z_{i',j'}) = \begin{cases} 1 &\mbox{if } i=i', j=j', \text{ and } k=k', \\
0 & \mbox{otherwise.} \end{cases}$$

\subsection{Main result when $\rho$ is close to $0$}

Our main result, which we will state below, gives a way to use an algorithm for almost any bilinear problem $T$ to solve the light bulb problem, even if $T$ only computes some of the terms of matrix multiplication, and if $T$ also computes other `noise' terms. To state our result, we need to define two relevant properties of $T$. The first property, the \textit{rank} of $T$, is a standard measure of how complicated $T$ is, while the another property, \textit{efficacy}, is the property we introduce for measuring how useful $T$ is for solving the light bulb problem.

\paragraph{Rank.} A tensor $T$ has rank $1$ if it can be written in the form
$$T = \left( \sum_{i \in [q_i], k \in [q_k]} \alpha_{i,k} \X_{i,k} \right) \left( \sum_{j \in [q_j], k \in [q_k]} \beta_{j,k} \Y_{j,k} \right) \left( \sum_{i \in [q_i], j \in [q_j]} \gamma_{i,j} \Z_{i,j} \right)$$
for coefficients $\alpha_{i,k}, \beta_{j,k}, \gamma_{i,j} \in \R$. More generally, the rank of $T$, denoted $\rank(T)$, is the minimum nonnegative integer $k$ such that there are rank $1$ tensors $T_1, \ldots, T_k$ with $T = T_1 + \cdots + T_k$.
Rank is the most prominent measure of the complexity of a tensor, and rank upper bounds for tensors yield algorithms for applying that tensor to matrices. 
For instance, Strassen~\cite{strassen} famously showed that the rank of the tensor $\langle 2,2,2 \rangle$ for multiplying two $2 \times 2$ matrices is at most $7$, and hence that one can multiply $n \times n$ matrices in time $O(n^{\log_2 7})$.

\paragraph{Efficacy.} The second property of $T$, which is a new property we introduce, is its \emph{efficacy}\footnote{We were inspired to pick this name by the `luminous efficacy' of a light bulb, which measures the ratio of how much light is produced and how much power is consumed. It bears similarity to other known statistical ratios like the `standardized second moment' and the `Fano factor'.}. 
For $i \in [q_i]$ and $j \in [q_j]$, the efficacy of $T$ at $(i,j)$ is given by:
$$\eff_{i,j}(T) := \frac{\sum_{k \in [q_k]} T(\X_{i,k} \Y_{j,k} \Z_{i,j})}{\sqrt{\sum_{i' \in [q_i], j' \in [q_j],k,k' \in [q_k]} T(\X_{i',k} \Y_{j',k'} \Z_{i,j})^2}}.$$
The numerator of $\eff_{i,j}(T)$ is the sum of the coefficients of all the entries which are supposed to be included in $Z_{i,j}$ in regular matrix multiplication. The denominator is the $\ell_2$ norm of the vector of coefficients of \emph{all} the terms which are included in $Z_{i,j}$ in $T$. Hence, one can think of $\eff_{i,j}(T)$ as a ratio of the `signal' and the `noise' of $T$ for computing the $(i,j)$ output entry of matrix multiplication.

Then, the efficacy of the whole tensor $T$ is the $\ell_2$ norm of the efficacies of all its output entries:
$$\eff(T) := \sqrt{\sum_{i \in [q_i], j \in [q_j]} \left(\eff_{i,j}(T)\right)^2}.$$
We will see that $\eff(T)$ measures how useful $T$ is for solving the light bulb problem from our main result, which shows how one could improve on the current best exponent $2 \omega / 3$:

\begin{theorem} \label{thm:mainintro}
For any tensor $T$, if $$\frac{\log(\rank(T))}{\log(\eff(T))} < \frac{2 \omega}{3},$$ then $$\omega_\ell < \frac{2 \omega}{3}.$$
Moreover, if $T$ has negligible aggregation time (see Section~\ref{subsec:aggregationintro} below), then $$\omega_\ell < \frac{\log(\rank(T))}{\log(\eff(T))}.$$
\end{theorem}

Hence, as long as $T$ is easy to compute ($\rank(T)$ is small) and it has a high ratio of signal to noise for computing matrix multiplication ($\eff(T)$ is large), one can use it to design a fast algorithm for the light bulb computation problem. (Again, this algorithm works with this exponent for any constant $\rho>0$, no matter how small.) We will see shortly that the algorithm consists of applying $T$ to pairs of carefully-chosen (but simple to construct) matrices, and then doing a simple analysis of the result. In other words, the algorithm itself is fairly simple, but the proof of correctness is quite involved.

\subsection{Aggregation time} \label{subsec:aggregationintro}

The aggregation time assumption in Theorem~\ref{thm:mainintro} relates to the initial aggregation step that appears in all prior matrix multiplication-based light bulb algorithms including ours~\cite{valiant2012finding,karppa2018faster,alman2018illuminating}. In~\cite{valiant2012finding}, aggregation took a significant amount of time which needed to be ``traded off'' against later steps of the algorithm. \cite{karppa2018faster} substantially improved aggregation to take a negligible amount of time compared to the rest of the algorithm.

The same technique of~\cite{karppa2018faster} applies in our setting as well, which makes aggregation negligible for all the tensors we study. We believe this technique makes the aggregation time negligible for all possible tensors $T$, although we're unable to prove this\footnote{Tensors with nonnegligible aggregation time would have very low efficacy, so that very long vectors are needed in our algorithm, but extremely low rank so that the algorithm may still be fast; see Appendix~\ref{sec:agg} below for more details.}.

Nonetheless, we prove in Theorem~\ref{thm:mainintro} that in order to improve the current best exponent $2 \omega / 3$, it suffices to find any tensor $T$ with $\frac{\log(\rank(T))}{\log(\eff(T))} < \frac{2 \omega}{3}$, ignoring the aggregation time condition.
To prove this, we show in Appendix~\ref{sec:agg} below that for any tensor $T$ with nonnegligible aggregation time, if the quantity $\frac{\log(\rank(T))}{\log(\eff(T))}$ is less than the current best exponent $2 \omega / 3$, then one can slightly modify $T$ to get a new tensor $T'$ with negligible aggregation time which still has $\frac{\log(\rank(T'))}{\log(\eff(T'))} <  2 \omega / 3$.

\subsection{Applying to Matrix Multiplication Generalizations} \label{subsec:tensorsintro}

To demonstrate the promise of Theorem~\ref{thm:mainintro}, we focus on tensors for $2 \times 2$ input matrices. (Using the notation above, we focus on $q_i = q_j = q_k = 2$.) This is the size of the tensor for Strassen's algorithm. As we discuss shortly, we introduce a new tensor with rank only $5$ which is able to to achieve a better exponent than Strassen's algorithm. 

We focus on this case for three reasons. First, using such a small tensor is typically necessary to design a practical algorithm (see, e.g., the introductions of~\cite{huang2016strassen,karppa2019probabilistic,pan2018fast,fawzi2022discovering} where practicality concerns are discussed). Second, using such small tensors lets us more concretely see how more general tensors can be used, especially in conjunction with locality-sensitive hashing later. Third, small tensors are typically a good test bed for further improvements based on larger tensors. (We discuss this in more detail shortly, in Section~\ref{sec:priorwork} below.)

\paragraph{Three Matrix Multiplication Generalizations.} We apply Theorem~\ref{thm:mainintro} to three tensors of interest. We will see that applying it when $T$ is a matrix multiplication tensor recovers the best known bound $\omega_\ell \leq \frac23 \omega$ of~\cite{karppa2018faster}, but that other tensors can yield even faster algorithms, including a new tensor we introduce.
See Figure~\ref{fig:tensors} for descriptions of the three tensors; their precise definitions and rank expressions are given in Section~\ref{sec:newtensor} below.

\begin{figure}[h!]
\begin{center} 
\begin{tabular}{ c|c|c|c| } 
Tensor Name & Rank & Tensor & \parbox{3cm}{Table of $\eff_{i,j}$\footnote{The indices $i$ and $j$ are in bold on the left and top of the table, respectively. The efficacy of the whole tensor is given below the table.}\\ value of $\eff$, and\\resulting $\omega_\ell$ bound} \\ \hline
\parbox{1.5cm}{$\langle 2,2,2 \rangle$ \\ (Strassen's algorithm \cite{strassen})} & 
$7$ & \parbox{4.5cm}{$(\X_{1,1} \Y_{1,1} + \X_{1,2}\Y_{2,1})\Z_{1,1} \\ + (\X_{1,1}\Y_{1,2} + \X_{1,2}\Y_{2,2})\Z_{2,1} \\+ (\X_{2,1}\Y_{1,1} + \X_{2,2}\Y_{2,1})\Z_{1,2} \\ + (\X_{2,1}\Y_{1,2} + \X_{2,2}\Y_{2,2})\Z_{2,2}$} & 
\parbox{4cm}{\footnotesize{\begin{center}\begin{tabular}{ l|l|l| } 
 \textbf{$i$\textbackslash$j$} & \textbf{1} & \textbf{2}  \\ 
 \hline
 \textbf{1} & $\sqrt{2}$ & $\sqrt{2}$  \\ 
 \hline
 \textbf{2} & $\sqrt{2}$ & $\sqrt{2}$  \\ 
 \hline
\end{tabular}
\\
\vspace{10pt}
$\eff(\langle 2,2,2 \rangle) = \sqrt{8}$
\\
\vspace{10pt}
$\omega_\ell \leq \frac{\log(7)}{\log(\sqrt{8})} < 1.872$ \vspace{-10pt}
\end{center}}}

\\ \hline
\parbox{1.5cm}{$SW$ \\ (Strassen-Winograd identity \cite{winograd1971multiplication})} & 
$6$ & \parbox{4.5cm}{$( \X_{1,2}\Y_{2,1})\Z_{1,1} \\ + (\X_{1,1}\Y_{1,2} + \X_{1,2}\Y_{2,2})\Z_{2,1} \\+ (\X_{2,1}\Y_{1,1} + \X_{2,2}\Y_{2,1})\Z_{1,2} \\ + (\X_{2,1}\Y_{1,2} + \X_{2,2}\Y_{2,2})\Z_{2,2}$} & 
\parbox{4cm}{\footnotesize{\begin{center}\begin{tabular}{ l|l|l| } 
 \textbf{$i$\textbackslash$j$} & \textbf{1} & \textbf{2}  \\ 
 \hline
 \textbf{1} & 1 & $\sqrt{2}$  \\ 
 \hline
 \textbf{2} & $\sqrt{2}$ & $\sqrt{2}$  \\ 
 \hline
\end{tabular}
\\
\vspace{10pt}
$\eff(SW) = \sqrt{7}$
\\
\vspace{10pt}
$\omega_\ell \leq \frac{\log(6)}{\log(\sqrt{7})} < 1.842$ \vspace{-10pt}
\end{center}}}

\\ \hline
\parbox{2cm}{$T_{2112}$ \\ (new tensor)} & 
$5$ &
\parbox{4.5cm}{\footnotesize{$(\X_{1,1}\Y_{1,1}+ \X_{1,2}\Y_{2,1} + O(\eps))\Z_{1,1} \\ + (\X_{1,2}\Y_{2,2} + O(\eps))\Z_{2,1} \\+ (\X_{2,1}\Y_{1,1} + O(\eps))\Z_{1,2} +\\  (\X_{2,1}\Y_{1,2} + \X_{2,2}\Y_{2,2} + O(\eps))\Z_{2,2}$\vspace{10pt}} 
 ($O(\eps)$ hides 
 arbitrarily small positive coefficients in terms of a parameter $\eps>0$)} &
\parbox{5cm}{\begin{center}\footnotesize{\begin{tabular}{ l|l|l| } 
 \textbf{$i$\textbackslash$j$} & \textbf{1} & \textbf{2}  \\ 
 \hline
 \textbf{1} & $\sqrt{2} - O(\eps^2)$ & $1 - O(\eps^2)$  \\ 
 \hline
 \textbf{2} & $1 - O(\eps^2)$ & $\sqrt{2} - O(\eps^2)$  \\ 
 \hline
\end{tabular}
\\
\vspace{10pt}
$\eff(T_{2112}) = \sqrt{6} - O(\eps^2)$}
\\
\vspace{10pt}
$\omega_\ell \leq \frac{\log(5)}{\log(\sqrt{6} - O(\eps^2))} \rightarrow 1.797$
\end{center}}

\\ \hline
\end{tabular}
\caption{Three tensors with $q=2$ which we use in our algorithm, along with the resulting bounds on $\omega_\ell$ from using them in Theorem~\ref{thm:mainintro}. See Section~\ref{sec:newtensor} below where we define these tensors exactly (without hiding terms in `$O(\eps)$') and give their rank expressions. The tensors $\langle 2,2,2 \rangle$ and $SW$ come from classical work on optimizing Strassen's algorithm, while $T_{2112}$ and its rank upper bound are both new.} \label{fig:tensors}
\end{center}
\end{figure}

First is the tensor for $2 \times 2$ matrix multiplication (denoted $\langle 2,2,2 \rangle$). We can calculate that $\eff(\langle 2,2,2 \rangle) = \sqrt{8}$, and hence, using Strassen's bound $\rank(\langle 2,2,2 \rangle) \leq 7$, that $\omega_\ell \leq 1.872$. Prior to this work, this was the smallest known exponent for the light bulb problem based on a $q=2$ tensor. Note that more generally, $\eff(\langle n,n,n \rangle) = n^{3/2}$, so applying Theorem~\ref{thm:mainintro} to the $n \times n$ matrix multiplication tensor (denoted $\langle n,n,n \rangle$) for large $n$ yields $\omega_\ell \leq \log(R(\langle n,n,n \rangle)) / \log(n^{1.5}) \leq \log(n^{\omega + o(1)}) / \log(n^{1.5}) \rightarrow \frac23 \omega$, which recovers the best known exponent for the light bulb problem~\cite{karppa2018faster}.

Second is the tensor $SW$, which consists of $7$ of the $8$ terms of $2 \times 2$ matrix multiplication, and which has rank $6$ via an identity by Winograd~\cite{winograd1971multiplication}. Recent work by Karppa and Kaski~\cite{karppa2019probabilistic} showed how to apply any tensor which consists of a subset of the terms of matrix multiplication to design a Boolean matrix multiplication algorithm. Follow-up work by Harris~\cite{harris2021improved} improved their analysis specifically for the tensor $SW$ to design a practical (since it is based on a small tensor) algorithm for Boolean matrix multiplication with exponent $\omega_B < 2.763$. It was previously unclear how to use $SW$, or any such `subset of matrix multiplication' tensor, to design an algorithm for the light bulb problem, or any problem which does not have a known reduction to Boolean matrix multiplication. Applying our Theorem~\ref{thm:mainintro} to $SW$ yields an algorithm with exponent $\omega_\ell < 1.842$, improving on Strassen's algorithm. More generally, for `subset of matrix multiplication' tensors, our bound on $\omega_\ell$ is strictly better than $\frac23$ times the bound on $\omega_B$ which Karppa and Kaski~\cite{karppa2019probabilistic} achieves, and equal to $\frac23$ times the bound on $\omega_B$ which Harris~\cite{harris2021improved} achieves (although~\cite{harris2021improved} only applies to some such tensors\footnote{It seems difficult to extend the approach of~\cite{harris2021improved} to `subset of matrix multiplication' tensors with very skewed patterns of terms, whereas~\cite{karppa2019probabilistic} applies to all such tensors. In this paper we use a technique to `regularize' the pattern of errors of a tensor (see Section~\ref{sec:symmetrizing} below) which seems at a first glance like it could apply to that setting as well, but unfortunately, the `space of errors' in that settings is 3-dimensional, whereas we critically use the fact that it is 2-dimensional here.}).
 
Third is a new rank-$5$ tensor $T_{2112}$ that we design and give a rank bound for in this paper. In terms of a parameter $\eps>0$, $T_{2112}$ is a sum of $6$ of the $8$ terms terms of $2 \times 2$ matrix multiplication, plus $7$ additional terms with coefficients $O(\eps)$, which can be made arbitrarily small by suitably picking $\eps$. (Note that one cannot eliminate these terms by setting $\eps=0$, since we divide by $\eps$ when showing that $T_{2112}$ has rank $5$; see Section~\ref{sec:newtensor} for more details.)

Prior work would have concluded by taking the limit $\eps \to 0$ in $T_{2112}$ that the `border rank' of 6 of the 8 terms of $2 \times 2$ matrix multiplication is $5$. Border rank bounds could be used instead of rank in conjunction with our approach by using the technique of Bini~\cite{bini1980border}. One advantage of Theorem~\ref{thm:mainintro} is that it allows one to plug constants $\eps>0$ into border rank expressions and avoid the complications of border rank. (Border rank identities are typically harder to find using numerical methods, and lead to less practical algorithms.) Setting just $\eps = 0.025$ suffices to get the best possible exponent using $q=2$:
\begin{theorem}
There is a tensor $T_{2112}$ with $q=2$ which achieves the exponent $\omega_\ell < 1.797.$
\end{theorem}

Before moving on, we note that a prior identity of Bini~\cite{bini1980relations} already gave a different tensor $B$ with (a different) 6 of the 8 terms of $2 \times 2$ matrix multiplication, and border rank $5$. However, our tensor $T_{2112}$ has one advantage over $B$.
That is, the pattern of $\eff_{i,j}(T_{2112})$, with larger entries along the diagonal, will allow us to use it in conjunction with hashing methods in our second result.

\subsection{Main result when $\rho$ is bounded away from $0$, using locality-sensitive hashing} \label{sec:general}

At a high level, our algorithm for Theorem~\ref{thm:mainintro} works by first mapping each of the $n$ different $x$ inputs into one of $q_i$ independently random buckets, and each of the $n$ different $y$ inputs into one of $q_j$ independently random buckets. (Recall that $q_i, q_j$ are two of the parameters defining the size of the tensor $T$; think of them as $n^{c}$ for a constant $0<c<1$, for instance by first taking an appropriate `Kronecker power' of $T$.) If the correlated pair were mapped into buckets $i \in [q_i]$ and $j \in [q_j]$, then our algorithm will succeed as long as $\eff_{i,j}(T)$ is large enough. The proof of Theorem~\ref{thm:mainintro} requires carefully balancing the parameters so that, when $i$ and $j$ are picked uniformly randomly, then this becomes fairly likely.

Our second main result shows how to improve Theorem~\ref{thm:mainintro} by combining it with one of the most prevalent techniques in nearest neighbor search: locality-sensitive hashing. 
The main idea to improve on this is to place the inputs into buckets using (a variation on) bit sampling locality-sensitive hashing, instead of uniformly random hashing. In this way, thinking of the buckets as $\{-1,1\}$ bit strings, we know that if the planted pair has correlation $\rho>0$, then they are likely put into buckets $i \in \{-1,1\}^{\log_2(q_i)}$ and $j \in \{-1,1\}^{\log_2(q_j)}$ which also have correlation close to $\rho$. If such buckets have larger $\eff_{i,j}(T)$ than uniformly random buckets, then we can speed up our algorithm.

By construction, our new tensor $T_{2112}$ has exactly this property! By renaming variables\footnote{Whenever $i$ or $j$ was $2$, we now call it $-1$.} and taking the limit $\eps \to 0$ for notational simplicity, we see that it has $\eff_{1,1}(T_{2112}) = \eff_{-1,-1}(T_{2112}) = \sqrt{2}$ and $\eff_{1,-1}(T_{2112}) = \eff_{-1,1}(T_{2112}) = 1$. More generally, once we've taken a Kronecker power so that $q_i = q_j = q$ is larger than $2$, the resulting tensor will have the property that, for buckets $i, j \in \{-1,1\}^{\log_2(q)}$ with correlation $\rho$, we have $\eff_{i,j}(T_{2112}^{\otimes \log_2(q)}) = 2^{(1 + \rho)(\log_2(q)/2)}$, whereas the median pair $i,j$ of buckets has only $\eff_{i,j}(T_{2112}^{\otimes \log_2(q)}) = 2^{\log_2(q)/2}$. Thus, in a sense, the efficacy of our tensor $T_{2112}$ is increasing with $\rho$, resulting in a faster algorithm. (We briefly note that although our analysis makes use of Kronecker powers of tensors, our algorithm itself does not, and only applies the tensor itself to input matrices.)

However, the formal statement of our result is more complicated than this because of a key detail behind Theorem~\ref{thm:mainintro} that we have thus far swept under the rug. Rather than map each input point into a single bucket, it actually makes many copies of each input point and independently maps them into buckets. This way, in order to solve the light bulb problem, it suffices for any one pair of copies of the correlated pair to map into buckets with large efficacy. Typically a locality-sensitive hashing scheme would map all the different copies of the same vector to the same bucket, and lose these savings. Nonetheless, we find a way to hash inputs into multiple buckets, so that the correlated pair is still hashed to correlated buckets, but the different pairs of buckets are `sufficiently independent' of each other so that whether or not each succeeds isn't too correlated.
Applying this to $T_{2112}$, we achieve:

\begin{theorem} \label{thm:intro2112}
For the tensor $T_{2112}$, the bound of Theorem~\ref{thm:mainintro} can be improved to
\begin{align*}
    \omega_{\ell} \leq \begin{cases}
    \frac{2\log 5}{\log\left(6(1-\rho)^{-\rho/2}(1+\rho)^{\rho/2}(1-\rho^2)^{1/2}\right) } & \textrm{ when~}\rho < 1/3,\\
    \frac{4\log 5}{(5+\rho)\log 2} & \textrm{ when~}1/3\leq \rho \leq 1.
    \end{cases}
\end{align*}
\end{theorem}

The resulting plot of $w_{\ell}$ with respect to $\rho$ from can be found in Figure~\ref{fig:pts} (in blue). Our Theorem~\ref{thm:mainintro}, as well as prior matrix multiplication-based algorithms for the light bulb problem, give the same running time exponent no matter how large $\rho>0$ is, whereas Theorem~\ref{thm:intro2112} uses hashing to improve with $\rho$. We also show Dubiner's algorithm, which is purely based on hashing, and is worse for small $\rho>0$, but better for larger $\rho$.

\begin{figure}
\centering
\begin{subfigure}{.48\textwidth}
  \centering
    \includegraphics[width=.9\linewidth]{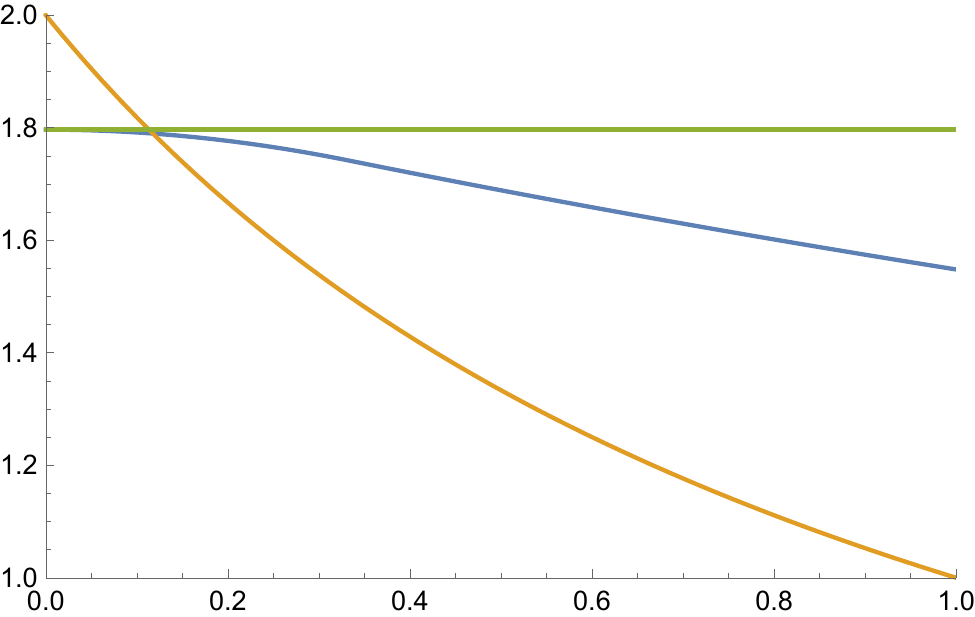}
    \caption{The running time exponent of three different methods in terms of $\rho \in [0,1]$.}
\end{subfigure}\hfill
\begin{subfigure}{.48\textwidth}
  \centering
  \includegraphics[width=.9\linewidth]{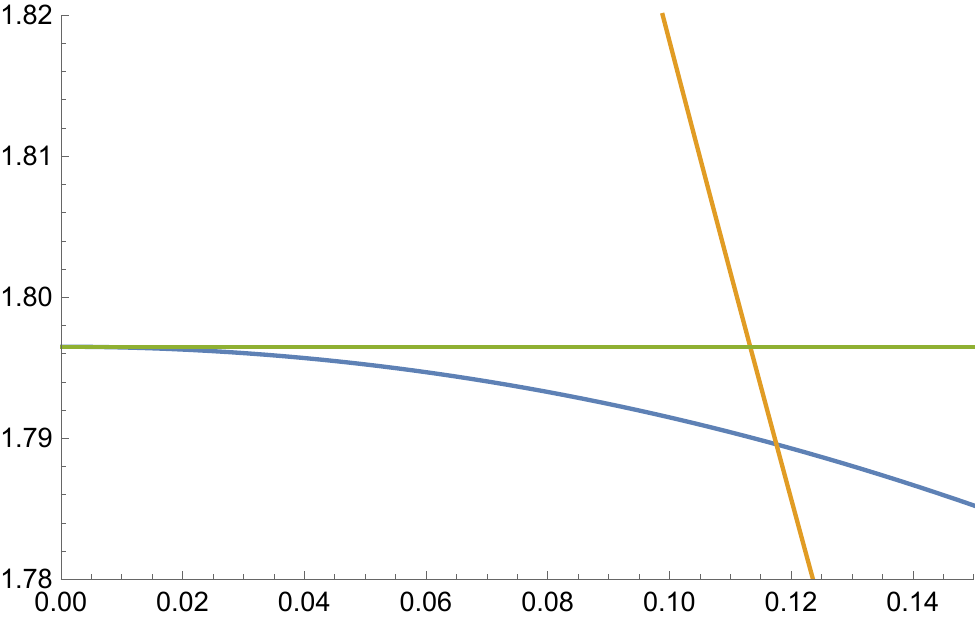}
  \caption{The same plot as in the figure to the left, but focused on the range $\rho\in [0,0.2]$.}
\end{subfigure}
\caption{The running time exponents ($y$-axis) of three algorithms in terms of $\rho$ ($x$-axis): Theorem~\ref{thm:intro2112} (blue line), Theorem~\ref{thm:mainintro} using $T_{2112}$ (green line), and Dubiner's algorithm  \cite{dubiner2010bucketing} (orange line).}
\label{fig:pts}
\end{figure}

\paragraph{Custom-tailored hash functions for other tensors.}

In order to prove Theorem~\ref{thm:intro2112}, we observed that bit sampling locality-sensitive hashing is likely to put the correlated pair of vectors into buckets $i,j$ where $i$ and $j$ are correlated, and hence have a higher-than-average value of $\eff_{i,j}(T_{2112}^{\otimes \log_2(q)})$. What if we are working with a different tensor $T$ for which the value of $\eff_{i,j}(T)$ does not increase as $i$ and $j$ are more correlated? Bit sampling locality-sensitive hashing won't give an improvement, but this is only one possible hash function. 

In fact, we can generalize Theorem~\ref{thm:intro2112} to almost any tensor. We show that for any $T$ whose efficacy matrix $[(\eff_{i,j}(T))^2]_{i,j}$ is not `degenerate' in some sense, one can custom-tailor hash functions for $T$ which result in an improved running time as $\rho$ grows.
The formal statement of this result is somewhat complicated; we defer the details to Section~\ref{sec:lsh} below. However, for one simple and important example, we show:

\begin{theorem} \label{thm:hashingintrogeneral}
Suppose $T$ is a $\langle q,q,q_k\rangle$-sized tensor which consists of a subset of the terms of a matrix multiplication tensor, and the matrix $[(\eff_{i,j}(T))^2]_{i\in[q],j\in[q]}$ has full rank. Let $\omega_\ell' $ be the exponent one would get from $T$ from applying Theorem~\ref{thm:mainintro}. Then, for every $\rho>0$, there is an $f(T,\rho) > 0$ such that the light bulb computation problem with correlation $\rho$ can be solved with the improved exponent $\omega_\ell' - f(T,\rho).$
\end{theorem}
Theorem~\ref{thm:hashingintrogeneral} shows that hashing can improve the algorithm based on almost any `subset of matrix multiplication' tensor. These are the same tensors used by Karppa and Kaski~\cite{karppa2019probabilistic} to solve Boolean matrix multiplication (they showed that bounds on their ranks give bounds on the `probabilistic rank' of matrix multiplication) and include the tensors other than $\langle 2,2,2 \rangle$ that we discussed above in Section~\ref{subsec:tensorsintro}. 

Intuitively, we require the matrix $[(\eff_{i,j}(T))^2]_{i,j}$ to have full rank in Theorem~\ref{thm:hashingintrogeneral} so that there are regions of buckets with higher efficacy that we could hope to hash the correlated pair to. For instance, if $T$ is a matrix multiplication tensor, then hashing cannot move the correlated pair to a better bucket since all buckets have the same efficacy, and indeed, the efficacy matrix has rank 1 since all its entries are equal.

\subsection{Comparison with Prior Work on Tensor and Nearest Neighbor Search Algorithms} \label{sec:priorwork}

\paragraph{Other Variants on Matrix Multiplication.} As mentioned at the beginning of Section~\ref{sec:intro}, prior work has solved Boolean matrix multiplication using tensors whose support is (a subset of) the support of matrix multiplication~\cite{cohn2013fast,karppa2019probabilistic,harris2021improved}. To our knowledge, we are the first to use tensors whose support may not be a subset of the support of matrix multiplication, and the first to use variants of matrix multiplication on a problem that is not (reducible to) Boolean matrix multiplication. This access to a larger class of tensors is what allows us to design a faster practical algorithm for the light bulb problem than the analogous fastest practical algorithm for Boolean matrix multiplication; our tensor $T_{2112}$ cannot be applied in the other settings (for fixed $\eps>0$). We hope our techniques could be used to apply these tensors to other problems in the future, particularly problems which are currently solved with exact matrix multiplication but which may only need approximate matrix multiplication.

\paragraph{Improving the asymptotic exponent}
We are optimistic that Theorem~\ref{thm:mainintro} can be used to improve the best known exponent for the light bulb problem by using larger tensors. Finding improvements based on larger tensors has historically been very difficult compared to finding improvements based on small tensors; for instance, it took almost 10 years after Strassen's algorithm based on a $2 \times 2$ identity before Pan~\cite{pan1978strassen} gave an improved exponent based on a larger tensor. Moreover, decades of work have gone into designing matrix multiplication algorithms for larger $q$ which we need to catch up to for the light bulb problem. Many of these techniques can be directly repurposed to the light bulb problem (for instance, it is not hard to prove a version of the asymptotic sum inequality~\cite{schonhage1981partial} in this setting), but the centerpiece of fast matrix multiplication algorithms, the Coppersmith-Winograd tensor, seems particularly designed for exact matrix multiplication, and it is not clear how to improve it for our light bulb setting. (More generally, it is a major open challenge to understand the effectiveness of the Coppersmith-Winograd tensor or find any useful variants on it~\cite{homs2022bounds,blaser2016degeneration,conner2019tensors,conner2022bad}.) Nonetheless, our computations suggest that this trend continues: our approach gives better bounds on $\frac32 \omega_\ell$ than prior approaches do for $\omega_B$ when restricted to certain small classes of tensors, such as tensors on the variable set of $\langle 3,3,3 \rangle$, or tensors on the variable set of $\langle 2,2,2 \rangle$ of rank at most $4$ (but none of these beats the bound of $T_{2112}$). We are optimistic improvements are possible for larger $q$ as well.

\paragraph{Exponent Comparison.}
Combining our result with~\cite{karppa2019probabilistic,harris2021improved} shows that, when restricted to rank bounds on small tensors over the same variable set as $\langle 2,2,2 \rangle$, the best known upper bounds have $\frac32 \omega_\ell < \omega_B < \omega.$
By comparison, the asymptotically best known upper bounds have $\frac32 \omega_\ell = \omega_B = \omega.$
It would be exciting to determine the relationship between these exponents in the asymptotic setting, perhaps using fine-grained reduction techniques. Indeed, although we know $\omega_B \leq \omega$, it's not clear in general what the relationship between $\omega_\ell$ and $\omega_B$ should be. Neither problem is known to be reducible to the other. Moreover, even our `efficacy' approach for the light bulb problem is incomparable to the `support rank'~\cite{cohn2013fast} and `probabilistic rank'~\cite{karppa2019probabilistic} approaches: our approach applies to a wider class of tensors, but the bound in Theorem~\ref{thm:mainintro} becomes worse when $T$ has large or negative coefficients, whereas `support rank' and `probabilistic rank' aren't impacted by what the coefficients are.
Nonetheless, from the small tensor regime, it appears plausible that $\omega_\ell < \omega_B$; it may be worth investigating whether other problems which are known to be reducible to matrix multiplication can actually be reduced to the light bulb problem instead!

Another advantage of our new algorithm is that it is not necessarily restricted to give exponents which are $\geq 4/3$. Recall that using the known bound $\omega_\ell \leq \frac23 \omega$, even if $\omega = 2$, one could only prove $\omega_\ell \leq 4/3$. Achieving an exponent less than $4/3$ requires another approach, and our Theorem~\ref{thm:mainintro} appears promising since it doesn't seem to have any such restrictions.

\paragraph{Tensors with Undesirable Terms.}  One motivation for this work is to use tensors with `undesirable' terms which, in other contexts, make them unusable. Typically one would expend rank to remove those terms, but one could design faster algorithms by allowing them to contribute to the efficacy of the tensor instead.
We've already discussed the case of tensors from border rank upper bounds via our example $T_{2112}$.
Tensors with undesirable terms also arise in the Laser method~\cite{strassenlaser1}, the tool used to design the best known upper bounds on $\omega$. A key step at the end of the Laser method, which was improved but not entirely removed in recent work of Alman and Vassilevska Williams~\cite{alman2021refined}, removes such undesirable terms. Leaving them in to contribute to the efficacy could lead to asymptotically faster algorithms.

\paragraph{Further Generalizations of the Light Bulb Problem.} In Section~\ref{sec:lsh} below, we show that our algorithm can also solve a generalization of the light bulb problem where each group of coordinates of the `correlated pair' are sampled from \emph{any} non-uniform joint distribution. Other generalizations have also been previously considered~\cite{valiant2012finding,karppa2018faster}, including a variant with many correlated pairs to find, and an `outlier correlation detection' variant where we are promised that the correlated pair has correlation $\rho$, and all other pairs have correlation at most $\tau$, for parameters $0 < \tau < \rho < 1$. Our approach can also solve these generalizations, by using the same techniques from prior work (which essentially amplify the differences in correlations by taking large Kronecker powers of the input vectors), since we focus on finding the correlated vectors after this amplification step. The details, which are essentially the same as in the past work, are omitted here. As discussed earlier, the best known algorithms for generalizations to other learning problems such as learning sparse parities or Juntas with noise also come from reductions to the light bulb problem~\cite{valiant2012finding}.

\paragraph{Other Closest Pair Problems.} The light bulb problem is an average case version of the bichromatic $(1+\eps)$-approximate closest pair problem, where one is given as input two sets $X,Y$ of $n$ points from a metric space, and one wants to find $x^* \in X$ and $y^* \in Y$ satisfying $\mathrm{dist}(x^*, y^*) \leq (1+\eps) \cdot \min_{x \in X, y \in Y} \mathrm{dist}(x,y)$. Similar to the previous state of the art for the light bulb problem, for many popular metric spaces, there are two known approaches for solving this problem: one based on matrix multiplication which is faster when $\eps>0$ is small~\cite{alman2015probabilistic,alman2016polynomial,alman2020faster}, and one based on locality-sensitive hashing which is faster when $\eps$ is larger~\cite{andoni2015optimal,andoni2017optimal}; see also \cite{andoni2018approximate}. To our knowledge, our hashing-based algorithm is the first to successfully combine matrix multiplication and hashing methods for any such problem. It is not hard to see that more straightforward ways to combine the two, such as hashing into smaller buckets and then using matrix multiplication within each bucket, cannot be faster than just using one of the two techniques on its own; we get around this by carefully choosing a hash function which correlates well with our chosen tensor. It would be exciting to apply a similar technique to other nearest neighbor search problems.

\subsection{Algorithm Overview}

Although Theorem~\ref{thm:mainintro} has a simple form (perhaps reminiscent of the bound $\omega \leq \log(\rank(\langle q,q,q \rangle)) / \log(q)$ which follows from the simple recursive argument), the algorithm itself involves a number of subtle steps in the case when $T$ is not `symmetric enough', and the proof of correctness is ultimately quite involved. At a high level, the elaborate probabilistic analyses which arise in prior works on the light bulb problem~\cite{valiant2012finding,karppa2018faster}, wherein one needs to prove tail bounds on sums of correlated events, return in full force when combined with errors which arise from using the tensor $T$ instead of matrix multiplication. We end up applying a simple variant on the Laser method to the tensor $T$ to `regularize' it without changing $\eff(T)$ too much, to help with the analysis.\footnote{Alman~\cite{alman2018illuminating} recently simplified some steps in prior algorithms for the light bulb problem using the polynomial method, but using our tensors in Alman's approach doesn't seem to work since Alman creates matrices with large entries to multiply.}

We focus here on describing the algorithm for Theorem~\ref{thm:mainintro} in the case when the tensor $T$ is sufficiently `symmetric' (as is the case for the three tensors described in Figure~\ref{fig:tensors}). Afterwards we will briefly discuss how we deal with asymmetric tensors, and how we extend our result using hashing to Theorems~\ref{thm:intro2112} and \ref{thm:hashingintrogeneral}.

\begin{algorithm}[h!]\caption{Light bulb algorithm for Theorem~\ref{thm:mainintro}}\label{alg:intro}
\begin{algorithmic}[1]
\Procedure{\textsc{LightBulb}}{$x_1\cdots x_n,y_1\cdots y_n\in \{-1,1\}^d,T$}
\State \Comment{$x_1,\cdots,x_n,y_1,\cdots,y_n$ are input vectors; $T$ is a tensor with $q_i=q_j=q_k=q$}
\State Let $N$ be such that $n^{\omega_{\ell}} = \rank(T)^{N}$. \Comment{$\omega_\ell = \frac{\log(\rank(T))}{\log(\eff(T))}$}
\State Calculate $\eff_{i,j}(T^{\otimes N})$ for all $i,j\in [q]^N$. \Comment{Can be done in negligible $\tilde{O}(q^{2N})$ time}
\State Let $g$ be such that $g^2\cdot |\{i,j\in[q]^N: \eff_{i,j}(T^{\otimes N})\geq g^2\}|$ is maximized, and set $t\leftarrow q^Ng/n$.\label{line:size_is_g}
\State $X_1,\cdots,X_{q^N},Y_1,\cdots,Y_{q^N} \leftarrow \emptyset$.
\For{$i\in [n]$}
\State Uniformly independently at random pick $i_1,\cdots,i_t$ and $j_1,\cdots,j_t$ from $[q^N]$. \label{line:step10}
\State Add $i$ to all the sets $X_{i_1},\cdots,X_{i_t},Y_{j_1},\cdots,Y_{j_t}$.
\EndFor
\State For all $i\in [q]^N$, let $a_i := \sum_{j\in X_i} x_j$ and $b_i:=\sum_{j\in Y_i} y_j$.\label{line:construct_a_b}
\State Let $A = [a_1^{\top},\cdots,a_{q^N}^{\top}]$ be $q^N\times d$ matrix of all $a_i^\top$ vectors.\label{line:construct_A}
\State Let $B = [b_1^{\top},\cdots,b_{q^N}^{\top}]$ be the $q^N \times d$ matrix of all $b_i^{\top}$ vectors. \Comment{Assume $d = q_k^N$.} \label{line:construct_B}
\State Random multiply each row of $A$ and $B$ by $-1$ or $1$.
\State Recursively apply $T$ to `multiply' $A$ and $B^{\top}$ and get their `product' $C$. \Comment{Take time $\tilde{O}(\rank(T)^N)$.}
\State Do the multiplication $100\log n$ times to get $C_1,\cdots,C_{100\log n}$. \Comment{Each time redo from beginning using fresh inputs.}
\State Find $(i,j)$ such that, there are at least $20\log n$ different $k\in [100\log n]$ that, $C_k[i,j] \geq 10\eff^{2}_{i,j}(T^{\otimes N}) \geq 10g^2$.
\EndProcedure
\end{algorithmic}
\end{algorithm}

The main algorithm is given in Algorithm~\ref{alg:intro}. 
There are two key results we need to prove its correctness. First, because of how $\eff_{i,j}$ is defined, if $|X_i| \cdot |Y_j| \leq \eff_{i,j}^2(T^{\otimes N})$ but $C[i,j] \geq \eff^2_{i,j}(T^{\otimes N})$, this means the correlated pair is likely to be in $X_i$ and $Y_j$. Roughly, we prove the fact that $\eff_{i,j}(T^{\otimes N})$ is so large \emph{means that} random noise cannot explain $C_k[i,j]$ being so large for too many $k$. Second, there is a decent probability that the bucketing used by the algorithm will result in copies of the correlated pair being put into $X_i$ and $Y_j$ for which $\eff_{i,j}(T^{\otimes N})$ is large enough.

This second result requires some work since whether or not the planted pair has been put into the pair of groups $(X_i, Y_j)$ is \emph{not} independent of whether it has been put into other pairs of groups. Moreover, it becomes more complicated in the case when the set $\{(i,j)\in[q]^N: \eff_{i,j}\geq g^2\}$ is `skewed' and mostly consists of pairs in a small number of rows or columns. In this case, we modify our algorithm by alternatingly applying either $T$ or its (appropriately defined) transpose. After this transformation, the large efficacies are `balanced enough' that a second moment method can be used to imply our second property.

Finally, as discussed in Section~\ref{sec:general} above, the idea behind Theorems~\ref{thm:intro2112} and \ref{thm:hashingintrogeneral} is to modify Line~\ref{line:step10} of the algorithm to sample the indices $i_1, \ldots, i_t$ according to a locality-sensitive hash function. This further complicates the analysis: not only are the pairs $(X_i, Y_j)$ which the planted pair has been put into not independent of each other, but even the buckets $X_{i_1}, \ldots, X_{i_t}$ which a single one of the planted vectors has been put into are not independent. We address this by independently perturbing $t$ copies of each input vector before applying a locality-sensitive hash function to them so that the different buckets are `sufficiently independent'. The key behind applying our approach to tensors $T$ other than $T_{2112}$ in Theorem~\ref{thm:hashingintrogeneral} is to do this perturbation in a biased way which correlates with the efficacy matrix of $T$. Fortunately, although the analysis requires these probabilistic analyses and explicitly analyzing the Kronecker power $T^{\otimes N}$, the algorithm itself is simple and only applies $T$ in the usual recursive way to pairs of matrices.

\subsection{Outline}
The remainder of our paper is organized as follows. After the preliminaries in Section~\ref{sec:prelims}, we prove Theorem~\ref{thm:mainintro} in Section~\ref{sec:firstalg}, then we prove Theorem~\ref{thm:intro2112} in Section~\ref{sec:lsh}.
In Section~\ref{sec:newtensor} we define and give the rank expressions for the tensors in Figure~\ref{fig:tensors}, including introducing our new tensor $T_{2112}$. In Section~\ref{sec:hashingworks} we prove Theorem~\ref{thm:hashingintrogeneral} that hashing can be used to improve the algorithm from most tensors. Finally, in Appendix~\ref{sec:agg} we discuss techniques from prior work for vector aggregation.

\section{Preliminaries}\label{sec:prelims}

\paragraph{Notation}

For positive integer $q$, write $[q] := \{1,2,3,\ldots,q\}$.

For an event $F$, write $[F]_{\mathbf{1}}$ to be $1$ if $F$ happens, and $0$ if $F$ does not happen.

For vectors $v \in \R^d$, and $\ell \in [d]$, we write $v[\ell]$ to denote entry $\ell$ of $v$. Similarly, for matrices $M \in \R^{d_1 \times d_2}$, and $\ell_1 \in [d_1], \ell_2 \in [d_2]$, we write $M[\ell_1, \ell_2]$ to denote the corresponding entry of $M$. 

\paragraph{Multinomial Coefficients}

If $a_1, a_2, \ldots, a_k \in [0,1]$ and $N \in \N$ are such that $\sum_{i=1}^k a_i = 1$, and $a_i \cdot N$ is an integer for all $i$, then we write the multinomial coefficient: $$\binom{N}{\{ a_i \cdot N \}_{i \in [k]}} := \prod_{i=1}^k \binom{N \cdot (1 - \sum_{j=1}^{i-1} a_j)}{a_i \cdot N}.$$ Standard bounds show that as $N \to \infty$, we have $$\binom{N}{\{ a_i \cdot N \}_{i \in [k]}} = \left( \prod_{i=1}^k a_i^{-a_i} \right)^{N - o(N)}.$$

\paragraph{Chebyshev's inequality}

Chebyshev's inequality says that if $U \in \R$ is a random variable with finite mean and finite non-zero variance, then for any real $k>0$ we have $$\Pr\left[|U - \E[U]| \geq k \cdot \sqrt{\var[U]}\right] \leq \frac{1}{k^2}.$$

\paragraph{Second Moment Method}

The second moment method says that if $U \in \R$ is a random variable such that $U$ is always nonnegative, and $\var[U]$ is finite, then $$Pr[U > 0] \geq \frac{(\E[U])^2}{\E[U^2]}.$$

\paragraph{Tensors}

For positive integer $q, q_k$, let $\X = \{\X_{i,k}\}_{i \in [q], k \in [q_k]}$ , $\Y = \{\Y_{j,k}\}_{j \in [q], k \in [q_k]}$, and $\Z = \{\Z_{i,j}\}_{i, j \in [q]}$. Most of the tensors in this paper will be over these sets, and we call a tensor over these sets a $\langle q,q,q_k\rangle$-sized tensor.

A tensor $T$ over $\X,\Y,\Z$ is a trilinear form in $\R^{|\X| \times |\Y| \times |\Z|}$. For $i,i',j,j' \in [q]$ and $k,k' \in [q_k]$ we write $T(\X_{i,k} \Y_{j,k'} \Z_{i',j'})$ for the coefficient of the term $\X_{i,k} \Y_{j,k'} \Z_{i',j'}$ in $T$. In other words, we can write: $$T = \sum_{i,i',j,j' \in [q], k,k' \in [q_k]} T(\X_{i,k} \Y_{j,k'} \Z_{i',j'}) \cdot \X_{i,k} \Y_{j,k'} \Z_{i',j'}.$$
We say $\X$ are the $x$-variables of $T$, $\Y$ are the $y$-variables of $T$, and $\Z$ are the $z$-variables of $T$.

The matrix multiplication tensor $\langle q,q,q_k \rangle$ is a tensor over $\X,\Y,\Z$ given by $$\langle q,q,q_k \rangle := \sum_{i,j \in [q], k \in [q_k]} \X_{i,k} \Y_{j,k} \Z_{i,j}.$$

\paragraph{Tensor Rank}

A tensor $T$ over $\X,\Y,\Z$ has rank $1$ if it can be written in the form $$T = \left( \sum_{i \in [q], k \in [q_k]} a_{i,k} \X_{i,k} \right)\left( \sum_{j \in [q], k \in [q_k]} b_{j,k} \Y_{j,k} \right)\left( \sum_{i,j \in [q]} c_{i,j} \Z_{i,j} \right)$$ for coefficients $a_{i,k}, b_{j,k}, c_{i,j} \in \R$. More generally, $\rank(T)$ is the minimum number of rank $1$ tensors whose sum is $T$.

\paragraph{Kronecker Product}

If $\X,\Y,\Z,\X',\Y',\Z'$ are sets of variables, $T$ is a tensor over $\X,\Y,\Z$, and $T'$ is a tensor over $\X',\Y',\Z'$, then the Kronecker product $T \otimes T'$ is a tensor over $\X \times \X', \Y \times \Y', \Z \times \Z'$ given by, for $x \in \X, x' \in \X', y \in \Y, y' \in \Y', z \in \Z, z' \in \Z'$, $$T \otimes T'((x,x')(y,y')(z,z')) = T(xyz) \cdot T'(x'y'z').$$

Notice in particular that for positive integers $q,q',q_k,q_k'$ we have $\langle q,q,q_k \rangle \otimes \langle q',q',q_k' \rangle = \langle qq', qq', q_k q_k' \rangle$. (Here, we say two tensors are equal if they are the same up to renaming variables.) We can view $\langle qq', qq', q_k q_k' \rangle$ as a tensor whose $\X$-variables are either $\{\X_{i,k}\}_{i \in [q \cdot q'], k \in [q_k \cdot q_k']}$ or $\{\X_{(i,i'),(k,k')}\}_{i \in [q], i' \in [q'], k \in [q_k], k' \in [q_k']}$. These are the same up to a natural bijection, and we will use both notations interchangeably.

For a tensor $T$ over $\X,\Y,\Z$ and positive integer $k$, we define the Kronecker power $T^{\otimes k}$ to be the Kronecker product of $k$ copies of $T$. It is a tensor over $\X^k,\Y^k,\Z^k$, and its coefficients are all the products of $k$ coefficients of $T$.

\paragraph{Applying a Tensor to Matrices}

If $T$ is a tensor over $\X,\Y,\Z$, and $A,B \in \R^{q \times q_k}$ are matrices, then the result of applying $T$ to $A$ and $B$ is a matrix $C \in \R^{q \times q}$ given by $$C[i,j] = \sum_{i',j' \in [q], k,k' \in [q_k]} T(\X_{i',k} \Y_{j',k'} \Z_{i,j}) \cdot A[i',k] \cdot B[j',k'].$$

The usual recursive algorithm (similar to Strassen's algorithm) shows that, for positive integers $N$, the tensor $T^{\otimes N}$ can be applied using only $\tilde{O}(\rank(T)^{N})$ field operations, or the improved bound $O(\rank(T)^{N})$ when $\rank(T) > q^2$.

\paragraph{Tensor Reflection} For a tensor $T$ over $\X,\Y,\Z$, its reflection $T^{\top}$ is another tensor over $\X,\Y,\Z$ given by, for $i,i',j,j' \in [q]$ and $k,k' \in [q_k]$, 
$$T^{\top}(\X_{i,k} \Y_{j,k'} \Z_{i',j'}) = T(\X_{j,k'} \Y_{i,k} \Z_{i',j'}).$$
This swaps the roles of the $\X$ and $\Y$ variables.

\paragraph{Kronecker Products of Matrices and Vectors} If $A \in \R^{n_a \times m_a}$ and $B \in \R^{n_b \times m_b}$ are matrices, one can analogously define their Kronecker product $A \otimes B \in \R^{n_a n_b \times m_a m_b}$ by, for $i \in [n_a], i' \in [n_b], j \in [m_a], j' \in [m_b]$, $A\otimes B[(i,i'), (j,j')] = A[i,j] \cdot B[i',j']$. Similarly, for vectors $u \in \R^{n_a}, v \in \R^{n_b}$, one can define $u \otimes v  \in \R^{n_a n_b}$ by $u \otimes v[(i,i')] = u[i] \cdot v[i']$.

Suppose $P$ is a property of vectors which is preserved under Kronecker product, i.e., if $u,v$ have the property, then so does $u \otimes v$. One example is the property of whether $\|v\|_2 \geq 1$. For a tensor $T$ over $\X,\Y,\Z$, let $S_P(T) \in \R^{q \times q}$ denote the matrix such that $S_P(T)[i,j] = 1$ if the vector $(T(\X_{i',k} \Y_{j',k'} \Z_{i,j}))_{i',j' \in [q], k,k' \in [q_k]}$ has property $P$, and $S_P(T)[i,j] = 0$ otherwise. Then, we can see that $S_P(T^{\otimes N}) = S_P(T)^{\otimes N}$. This will be particularly helpful to us in the case when $P$ is the property that $\eff_{i,j}(T) \geq f$ for some threshold $f$. (See Definition~\ref{def:eff} below for the the definition of $\eff$.)

\section{Algorithm for the light bulb problem} \label{sec:firstalg}

\begin{definition}[Efficacy] \label{def:eff}
Given any $\langle q,q,q_k \rangle$-sized tensor $T$, for $i,j \in [q]$, we define
the \emph{$(i,j)$-efficacy of $T$} as:
$$\eff_{i,j}(T) := \frac{ \sum_{k \in [q_k]} T(\X_{i,k} \Y_{j,k} \Z_{i,j}) }{ \sqrt{\sum_{i', j' \in [q], k,k' \in [q_k]} T(\X_{i',k} \Y_{j',k'} \Z_{i,j})^2} }.
$$

We further define the \emph{efficacy of $T$} as: $$\eff(T) := \sqrt{\sum_{i \in [q]} \sum_{j \in [q]} \left( \eff_{i,j}(T) \right)^2}.$$
\end{definition}

Note that if $T, T'$ are two tensors,
for $(i,i'), (j,j') \in [q]^2$, 
we have $\eff_{(i,i'), (j,j')}(T \otimes T') = \eff_{i,j}(T) \cdot \eff_{i',j'}(T').$

We now begin giving our algorithm for the light bulb problem. Our goal is to analyze Algorithm~\ref{alg:intro} in order to prove Theorem~\ref{thm:mainintro}. In particular, we assume throughout this section that $T$ is such that the aggregation step (lines~\ref{line:construct_a_b}, \ref{line:construct_A}, , \ref{line:construct_B}) take negligible time compared to the rest of the algorithm; in Appendix~\ref{sec:agg} below, we show how to modify $T$, if necessary, so that this is the case.

\begin{theorem} \label{thm:main1}
Suppose $T$ is a $\langle q, q, q_k\rangle$-sized tensor. For any $f \geq 1$, and any set $S_f \subseteq [q]^2$ such that $\eff_{i,j}(T) \geq f$ for all $(i,j) \in S_f$, we have $\omega_\ell \leq \log(\rank(T) \cdot q^2 / |S_f|)/\log(f \cdot q)$.
\end{theorem}

\begin{proof}
Suppose we are given as input $x_1, \ldots, x_n,$ $y_1, \ldots, y_n \in \{-1,1\}^d$ which are all generated independently and uniformly at random except for an unknown $(i^*, j^*) \in [n]^2$ with $\langle x_{i^*}, y_{j^*} \rangle \geq \rho \cdot d$. 
Permute the inputs at random so that $(i^*, j^*)$ is a uniformly random pair in $[n]^2$. 

We can solve the light bulb problem using $O(\log n)$ calls of its decision version -- we randomly take half $x$ and half $y$, and for the decision problem, we need to distinguish between two cases 1). all inputs are uniformly at random in $\{-1,1\}^d$, and 2). there exists one correlated $(i^*,j^*)$ pair.  From now on, we will only consider the decision version.

Let $m = (\frac{20 n}{\rho})^{\frac{1}{1 + \log_q(f)}}$ and let $g = n/m$. Partition $x_1, \ldots, x_n$ into $m$ groups $X_1, \ldots, X_m$ of size $g$ each, and partition $y_1, \ldots, y_n$ into $m$ groups $Y_1, \ldots, Y_m$ of size $g$ each. For each $i \in [m]$, create vectors $a_i, b_i \in \R^d$ given by $a_i = \sum_{u \in X_i} u$ and $b_i = \sum_{v \in Y_i} v$. (These vectors aggregate all the data points which were put into the same group; we will see soon that if groups $i$ and $j$ contain the correlated pair, then $a_i$ and $b_j$ are still somewhat correlated.)
Let $s_a,s_b\in \{-1,1\}^m$ be two vectors whose entries are independently uniformly sampled from $\{-1,1\}$.
Finally, we form the matrices $A, B \in \R^{m \times d}$ whose rows are $s_a[1] \cdot a_1, \ldots, s_a[m] \cdot a_m$ and $s_b[1] \cdot b_1, \ldots, s_b[m] \cdot b_m$, respectively.

For simplicity, let us assume that $d = m^{\log(q_k)/\log(q)}$ so that $T^{\otimes c}$ is a $\langle m,m,d\rangle$-sized tensor that can be used on $A$ and $B$, where $c = \log(m)/\log(q)$.
%\Hengjie{We mentioned aggregation here.}
As discussed in the introduction, if one would like to remove this requirement, then using the `compressed matrices' method introduced in~\cite[Section 4.2]{karppa2018faster}, one can `expand' lower-dimensional vectors, and hence relax this assumption to only require $d \geq \Omega(\log n)$ while only decreasing $\rho$ by a negligible factor\footnote{\label{footnote:aggregation}In fact, the result of the compressed matrices method gives that if $x_i, y_j$ are not the correlated pair, the the entries of the entry-wise product $x_i \circ y_j$ are only pairwise-independent of each other, and not fully independent. (This is because they are products of different entries of the original vectors.) As we will see below, this pairwise-independence suffices for our algorithm.

The compressed matrices technique particularly speeds up the time to compute the aggregated vectors $a_i, b_j$ so that it is negligible compared to the remaining running time of the algorithm, and one can confirm that it remains negligible here. We refer the reader to~\cite[Section 4.4]{karppa2018faster} for more details.}; see Appendix~\ref{sec:agg} below for more details.

We now apply the usual recursive algorithm using the tensor $T$ to the matrices $A$ and $B^T$, resulting in the matrix $C \in \R^{m \times m}$. The running time is $\tilde{O}(m^{\log(\rank(T))/\log(q)}) = \tilde{O}(n^{\log(\rank(T))/\log(qf)})$. The output $C$ is the result of applying the tensor $T^{\otimes c}$ to the matrices $X = A$ and $Y = B^T$, so that each entry $C[i,j]$ is the sum:
\begin{align}\label{eq:cij}
C[i,j] = \sum_{i_a, j_b\in [q^c], k_a, k_b \in [q_k^c]} T^{\otimes c}(X_{i_a,k_a} Y_{j_b,k_b} Z_{i,j}) \cdot A[i_a, k_a] \cdot B[j_b, k_b].
\end{align}

Consider the product of two terms $A[i,k_a] \cdot B[j,k_b]$. These are distributed as follows:
\begin{itemize}
    \item If there's a planted pair $(x_{i^*},y_{j^*})$ and $k_a = k_b$, $x_{i^*} \in X_i$, and $y_{j^*} \in Y_j$, then this is the sum of $g^2$ random $\{-1,1\}$ variables which are pairwise-independent from each other. They all have mean $0$ except one of them has mean $\rho$, so the entire variable $A[i,k_a] \cdot B[j,k_b]$ has mean $\rho$ and variance $g^2$.
    \item Otherwise, it is the sum of $g^2$ uniformly random pairwise-independent $\{-1,1\}$ variables, which has mean $0$ and variance $g^2$.
\end{itemize}

Let's compute the variance of $C[i,j]$. We showed earlier that each term $A[i_a, k_a] \cdot B[j_b, k_b]$ has variance $g^2$, so $T^{\otimes c}(X_{i_a,k_a} Y_{j_b,k_b} Z_{i,j}) \cdot A[i_a, k_a] \cdot B[j_b, k_b]$ has variance $T^{\otimes c}(X_{i_a,k_a} Y_{j_b,k_b} Z_{i,j})^2 \cdot g^2$. 
Since we use $s_a,s_b$ entry-wise independently sampled from $\{-1,1\}$, the terms in the sum (\ref{eq:cij}) are pairwise-independent of each other. It follows that regardless of whether there's a planted pair, every $C[i,j]$ has variance:

$$\var[C[i,j]] = g^2 \cdot \sum_{i_a, j_b \in [q^c], k_a, k_b \in [q_k^c]} T^{\otimes c}(X_{i_a,k_a} Y_{j_b,k_b} Z_{i,j})^2.$$

Next let's compute the mean of $C[i,j]$. If there's no planted pair, every $C[i,j]$ has mean $0$. If the planted pair $(x_{i^*},y_{j^*})$ exists, let's assume $x_{i^*}\in X_i$ and $y_{j^*} \in Y_j$ and only consider the mean of $C[i,j]$.
Recall that $A[i_a, k_a] \cdot B[j_b, k_b]$ has mean nonzero only if $k_a = k_b$, $i_a = i$ and $j_b = j$. It follows by linearity of expectation that
\begin{align*}
\E[C[i,j]] = & ~  \sum_{k \in [q_k^c]} T^{\otimes c}(X_{i,k} Y_{j,k} Z_{i,j}) \cdot \E[A[i, k] \cdot B[j, k]]\\
= & ~ \rho \cdot s_a[i]\cdot s_b[j]\cdot \sum_{k \in [q_k^c]} T^{\otimes c}(X_{i,k} Y_{j,k} Z_{i,j}).
\end{align*}

To summarize: when there's no planted pair, every $C[i,j]$ has mean $0$ and variance
$$g^2 \cdot \sum_{i_a, j_b \in [q^c], k_a, k_b \in [q_k^c]} T^{\otimes c}(X_{i_a,k_a} Y_{j_b,k_b} Z_{i,j})^2.$$ 
When the planted pair exists, and $x_{i^*}$ is in $X_i$, $y_{j^*}$ is in $Y_j$, and the entry $(i,j)$ is `good' (We say $(i,j)$ is `good' if $\eff_{i,j}(T^{\otimes c}) \geq f^c$), the ratio of its mean and standard deviation is at least:

\begin{align*}
& ~ \left| \frac{ \rho \cdot s_a[i]\cdot s_b[j]\cdot \sum_{k \in [q_k^c]} T^{\otimes c}(X_{i,k} Y_{j,k} Z_{i,j}) }{ g^2 \cdot \sum_{i_a, j_b \in [q^c], k_a, k_b \in [q_k^c]} T^{\otimes c}(X_{i_a,k_a} Y_{j_b,k_b} Z_{i,j})^2 }  \right| \\
= & ~  \frac{\rho}{g} \cdot \eff_{i,j}(T^{\otimes c}) \\
\geq & ~  \frac{\rho}{g} \cdot f^c \\
= & ~ \frac{\rho \cdot m}{n} \cdot f^c \\
\geq & ~ \frac{\rho \cdot m}{n} \cdot f^{\log(m)/\log(k)} \\ 
= & ~ \frac{\rho}{n} \cdot m^{1 + \log(f)/\log(k)}\\ 
= & ~ \frac{\rho}{n} \cdot \frac{20 \cdot n}{\rho} \\ 
= & ~  20.
\end{align*}

Therefore, follows by Chebyshev's inequality, when there's no planted pair, $C[i,j] \leq  10 \var[C[i,j]]$ with probability $\geq 0.99$ for all $(i,j)\in [m]^2$, and when there exists a planted pair in a `good' entry $(i^*,j^*)$, then $C[i^*,j^*] \geq 10 \var[C[i^*,j^*]]$ with probability $\geq 0.99$. 
Thus, we can independently repeating $O(\log n)$ times to distinguish the two cases with polynomially-low error. 

Let $|S_f| = q^{\alpha}$. Since $\eff_{(i,i'), (j,j')}(T \otimes T') = \eff_{i,j}(T) \cdot \eff_{i',j'}(T')$, there are at least $|S_f|^c$ pairs of $(i,j)\in [m]^2$ that $\eff_{i,j}(T^{\otimes c}) \geq f^c$. So the planted pair has $|S_f|^c/q^{2c} = q^{-(2-\alpha)c}$ probability going to a `good' entry. We repeat $q^{(2-\alpha)c} \log n$ times to make sure the planted pair goes to a `good' entry at least once with high probability, therefore we can distinguish the planted-pair case and the non planted-pair case.

\textbf{Running time:}
Each run cost $\tilde{O}(n^{\log(\rank(T))/\log(qf)})$ time.
We repeat the whole procedure $O(n^{(2-\alpha) \log(q) / \log(qf)} \cdot \log^2 n)$ times to succeed with high probability. 

The total running time is $\tilde{O}(n^{\log(\rank(T))/\log(qf)} \cdot n^{(2-\alpha) \log(q) / \log(qf)}) = \tilde{O}(n^{\log(\rank(T) q^{2-\alpha})/\log(qf)}) = \tilde{O}(n^{\log(\rank(T) q^{2}/|S_f|)/\log(qf)})$, as desired.
\end{proof}

\subsection{Improvement when $S_f$ is not `skewed'}\label{sec:imp1}

We next show that in the special case when $S_f$ is not too `skewed', we can improve the bound of Theorem~\ref{thm:main1}.

Recall that for $f \geq 1$, we chose a subset $S_f \subseteq [q]^2$ consisting of pairs $(i,j) \in [q]^2$ for which $\eff_{i,j}(T) \geq f$. Let's define the following measurement of how a set $S_f$ is closed to 'skewed'.

\begin{definition}[$V_x(S)$ and $V_y(S)$]\label{def:v_x_v_y}
For any set $S\subseteq  [q]^2$, let's define $V_{x}(S) := \sum_{i \in [q]} |\{ j \in [q] \mid (i,j) \in S \}|^2$, and similarly define $V_{y}(S) := \sum_{j \in [q]} |\{ i \in [q] \mid (i,j) \in S \}|^2$.
\end{definition}

\begin{theorem} \label{thm:main2}
Suppose $T$ is a $\langle q, q, q_k\rangle$-sized tensor. For any $f \geq 1$, and any set $S_f \subseteq [q]^2$ such that $\eff_{i,j}(T) \geq f$ for all $(i,j) \in S_f$, if $V_{x}(S_f), V_{y}(S_f) \leq |S_f|^{1.5}$, then, $\omega_\ell \leq \log(\rank(T)) / \log(f \cdot \sqrt{|S_f|})$.
\end{theorem}

\begin{proof}
Let $E_1 = \log(\rank(T))/\log(qf)$ and $E_2 = \log(q^2 / |S_f|) / \log(qf)$.

Recall that in Theorem~\ref{thm:main1}, we randomly partitioned the inputs $x_1, \ldots, x_n$ into sets $X_1, \ldots, X_m$, and the inputs $y_1, \ldots, y_n$ into sets $Y_1, \ldots, Y_m$, then we ran an algorithm which takes time $\tilde{O}(n^{E_1})$, and which will (with high probability) distinguish the 
all-random case and the planted-pair case, if $x_{i^*}$ was put into set $X_i$ and $y_{j^*}$ was put into set $Y_j$ such that $(i,j) \in S_f^{\otimes c}$. Call such $(i,j)$ `good'. The probability $(i,j)$ is good is $(|S_f|^c / q^{2c}) = n^{-E_2}$.

In Theorem~\ref{thm:main1}, we then repeated $\tilde{O}(n^{E_2})$ times, resulting in a final running time of $\tilde{O}(n^{E_1 + E_2})$, but we will now instead do something more clever. 

Let $t = n^{E_2 / (2 - E_2)}$. We will make $t$ copies of each $x_i$ and $y_j$ vector, and then run the above algorithm on this new instance with $n \cdot t$ vectors, with the caveat that we never put two copies of the same $x$ vector in the same group $X_i$, or two copies of the same $y$ vector in the same group $Y_j$. The running time of this is $\tilde{O}((nt)^{E_1}) = \tilde{O}(n^{2 \cdot E_1 / (2 - E_2)})$, and the probability that a particular correlated pair will be put in a good pair of groups is now $(nt)^{-E_2}$. Since we made $t$ copies of $x_{i^*}$ and $t$ copies of $y_{j^*}$, there are now $t^2$ correlated pairs, so the expected number of correlated pairs in a good pair of groups is $$t^2 \cdot (nt)^{-E_2}  = t^{2 - E_2} \cdot n^{- E_2} = n^{\frac{E_2}{2-E_2} \cdot (2 - E_2) - E_2} = n^{E_2 - E_2} = 1.$$ With more work, and using our bounds on $V_{x}(S_f)$ and $V_{y}(S_f)$, we can show that there is a positive constant probability that a correlated pair was put in a good pair of groups. See Lemma~\ref{lem:combrect} below for the details. Hence, if we repeat $O(\log n)$ times, a correlated pair will be put in a good pair of groups with polynomially low error.

The proof that our algorithm can distinguish the planted-pair case with the non planted-pair case is almost identical to the proof in Theorem~\ref{thm:main1}, except the following. 
In the old proof, when the planted pair is in $X_i$ and $Y_j$ and $\eff_{i,j}(T^{\otimes c}) \geq f^c$, there's $\geq 0.99$ probability that $C[i,j] \geq 10 \var[C[i,j]]$. Now, we can only prove the probability is $\geq 0.24$ because duplicated vectors created correlation. Nonetheless, we can distinguish from the non-planted-pair case, in which every $C[i,j] \geq 10\var[C[i,j]]$ with probability $\leq 0.01$. We show probability $\geq 0.24$ as follows.

Let 
\[
P[i_a,j_b]:=\sum_{k_a, k_b \in [d]} T^{\otimes c}(X_{i_a,k_a} Y_{j_b,k_b} Z_{i,j}) \cdot a_{i_a, k_a} \cdot b_{j_b, k_b},
\]
so that $C[i,j]$ can be written as
\begin{align}\label{eq:cij_new}
C[i,j] = \sum_{i_a, j_b\in [m]} P[i_a,j_b] \cdot s_a[i_a]\cdot s_b[j_b],
\end{align}
where $s_a,s_b\in\{-1,1\}^m$ has i.i.d. $\{-1,1\}$ entries.

Use Lemma~\ref{lem:random_-1_1} on $C[i,j]$, we conclude that with $\geq 1/4$ probability over random choice of $s_a$, $s_b$, $|C[i,j]| \geq |P[i,j]|$.

$\E[P[i,j]] = \rho \cdot \sum_{k \in [q_k^c]} T^{\otimes c}(X_{i,k} Y_{j,k} Z_{i,j})$ and $\var[P[i,j]] \leq \var[C[i,j]]$. By the same analysis in Theorem~\ref{thm:main1}, with $\geq 1/4-0.01$ probability, $|C[i,j]| \geq 
|\E[P[i,j]]| - 10\var[P[i,j]] \geq 10\var[C[i,j]]$.

The resulting exponent is hence:
$$2 \cdot E_1 / (2 - E_2) = 2 \cdot E_1 / (\log(f^2 \cdot |S_f|) / \log(qf)) = \log(\rank(T)) / \log(f \cdot \sqrt{|S_f|}).$$
\end{proof}

\begin{lemma}\label{lem:random_-1_1}
Let $P \in \R^{n\times n}$ be a matrix. Let $a,b\in \{-1,1\}^n$ be entry-wise i.i.d. uniformly sampled from $\{-1,1\}$. Then with probability $\geq 1/4$, 
\begin{align}\label{eq:sum1}
\left| \sum_{i,j\in[n]} a[i]\cdot b[j]\cdot P[i,j] \right| \geq \left| P[1,1] \right|.
\end{align}
\end{lemma}
\begin{proof}
Define $\overline{T} = \sum_{j\geq 2} P[1,j]\cdot b[j]$, and $T = \overline{T} + P[1,1]\cdot b[1]$. 
Regardless of how large $\overline{T}$ is, with $\geq 1/2$ probability over $b[1]$, $|T| \geq |P[1,1]|$.

Let $S$ be the left side of Eq.~\eqref{eq:sum1}. Notice that
\[
S=\left|\sum_{i \geq 2} a[i](\sum_{j\in[n]} P[i,j]\cdot b[j]) + T\cdot a[1] \right|,
\]
use the same argument, we have $|S|\geq  |T| \geq |P[1,1]|$ with $\geq 1/4$ probability.
\end{proof}

\subsection{Probabilistic lemma for when $S_f$ is not `skewed'} \label{subsec:skew}

\begin{lemma}\label{lem:combrect2}
Suppose $q$ is a positive integer and $S \subseteq [q]^2$ is a nonempty subset. Let $V_{x}(S) := \sum_{i \in [q]} |\{ j \in [q] \mid (i,j) \in S \}|^2$ and $V_{y}(S) := \sum_{j \in [q]} |\{ i \in [q] \mid (i,j) \in S \}|^2$, and suppose that $V_x(S)  \leq |S|^{1.5}$ and $V_y(S)  \leq |S|^{1.5}$. 

Suppose we pick $S_x, S_y \subseteq [q]$ of size $|S_x|, |S_y| \geq q/\sqrt{|S|}$ independently and uniformly at random. Then, the probability that $|(S_x \times S_y) \cap S| > 0$ is at least $1/4$.
\end{lemma}

\begin{proof}
Let $U$ denote the random variable $|(S_x \times S_y) \cap S|$. We will use the second moment method, which says that $$\Pr[U>0] \geq \frac{(\E[U])^2}{\E[U^2]}.$$

First, by linearity of expectation, we compute that $\E[U] = |S_x| \cdot |S_y| \cdot (|S|/q^2) = 1$.

Next, again by linearity of expectation, we compute:

\begin{align*}
\E[U^2] =& \sum_{(i,j)\in S} \sum_{(i',j')\in S} \Pr[i, i' \in S_x \text{ and } j, j' \in S_y]
\\=& \sum_{(i,j)\in S} \Bigg( \left( \frac{1}{\sqrt{|S|}} \right)^2 + \left( \frac{1}{\sqrt{|S|}} \right)^3 \cdot |\{ (i',j') \in S \mid i=i' \text{ xor } j=j' \}| \\ &+ \left( \frac{1}{\sqrt{|S|}} \right)^4 \cdot |\{ (i',j') \in S \mid i\neq i' \text{ and } j\neq j' \}| \Bigg)
\\ =&  \left( \frac{1}{\sqrt{|S|}} \right)^2 \cdot |S| + \left( \frac{1}{\sqrt{|S|}} \right)^3 \cdot (V_x(S) + V_y(S) -2|S|) \\ &+ \left( \frac{1}{\sqrt{|S|}} \right)^4 \cdot (|S|^2 - V_x(S) - V_y(S) + |S|) 
\\ =&  2 + \frac{1}{|S|} - \frac{2}{|S|^{0.5}} + (V_x(S) + V_y(S)) \cdot \left( \frac{1}{|S|^{1.5}} - \frac{1}{|S|^{2}}\right)
\\ \leq&  4 + \frac{1}{|S|} - \frac{4}{|S|^{0.5}}
\\ <&  4,
\end{align*}
where the last step follows since $\frac{1}{s} - \frac{4}{\sqrt{s}} < 0$ for all $s \geq 1$.

In total, we get as desired that $$\Pr[U>0] \geq \frac{(\E[U])^2}{\E[U^2]} = \frac{1}{\E[U^2]} > \frac{1}{4}.$$
\end{proof}

\begin{lemma}\label{lem:combrect}
Suppose $q$ is a positive integer and $S \subseteq [q]^2$ is a nonempty subset. Let $V_{x}(S) := \sum_{i \in [q]} |\{ j \in [q] \mid (i,j) \in S \}|^2$ and $V_{y}(S) := \sum_{j \in [q]} |\{ i \in [q] \mid (i,j) \in S \}|^2$, and suppose that $V_x(S) + V_y(S) \leq |S|^{1.5}$. Let $c$ be a positive integer.

Suppose we pick $S_x, S_y \subseteq [q^c]$ of size $|S_x|, |S_y| \geq q^c/\sqrt{|S|^c}$ independently and uniformly at random. Then, the probability that $|(S_x \times S_y) \cap S^{\otimes c}| > 0$ is at least $1/4$.
\end{lemma}

\begin{proof}
Apply Lemma~\ref{lem:combrect2} to $S^{\otimes c} \subseteq [q^c]^2$. Note that $V_x(S^{\otimes c}) = (V_x(S))^c$ and $V_y(S^{\otimes c}) = (V_y(S))^c$, so the conditions are still satisfied after taking the $c^{th}$ Kronecker power.
\end{proof}

\subsection{Symmetrizing a tensor to avoid skew} \label{sec:symmetrizing}

\begin{definition}
For a positive integer $q$, we say a set $S \subseteq [q]^2$ is \emph{regular} if there are positive integers $a$ and $b$ such that, for all $i \in [q]$, $|\{ j \in [q] \mid (i,j) \in S \}|$ is either equal to $a$ or equal to $0$, and similarly for all $j \in [q]$, $|\{ i \in [q] \mid (i,j) \in S \}|$ is either equal to $b$ or equal to $0$.
\end{definition}

\begin{lemma} \label{lem:symmetrize}
Suppose $q$ is a positive integer and $S \subseteq [q]^2$ is regular. Let $V_{x}(S) := \sum_{i \in [q]} |\{ j \in [q] \mid (i,j) \in S \}|^2$ and $V_{y}(S) := \sum_{j \in [q]} |\{ i \in [q] \mid (i,j) \in S \}|^2$.
Then, $V_x(S) \cdot V_y(S) \leq |S|^3$.
\end{lemma}

\begin{proof}
If $S$ is empty, then the result holds since $V_x(S) = V_y(S) = |S| = 0$. Otherwise, assume without loss of generality that $(1,1) \in S$.

Let $a = |\{ j \in [q] \mid (1,j) \in S \}|$, and $b = |\{ i \in [q] \mid (i,1) \in S \}|$. Since $S$ is regular, there are $|S|/a$ choices of $i \in [q]$ for which $|\{ j \in [q] \mid (i,j) \in S \}| = a$, and so $V_x(S) = \frac{|S|}{a} \cdot a^2 = a \cdot |S|$. Similarly, $V_y(S) = b \cdot |S|$, which means $V_x(S) \cdot V_y(S) = a \cdot b \cdot |S|^2$.

Let $W := \{ j \in [q] \mid (1,j) \in S \}$, so $|W| = a$. Next, for each $j \in W$, let $W_j := \{ (i,j) \mid i \in [q] \text{ and } (i,j) \in S \} \subseteq S$. Let $W' := \bigcup_{j \in W} W_j \subseteq S$, and note that the $W_j$ sets are disjoint, so $|W'| = \sum_{j \in W} |W_j|$. By definition of $W$, we know that $W_j$ is nonempty for each $j \in W$, and so $|W_j| = b$. It follows that $|S| \geq |W'| = |W| \cdot b = a \cdot b$. 

We thus get as desired that $V_x(S) \cdot V_y(S) = a \cdot b \cdot |S|^2 \leq |S|^3$.
\end{proof}

\begin{theorem}\label{thm:main3}
Suppose $T$ is a $\langle q, q, q_k\rangle$-sized tensor. For any $f \geq 1$, and any regular set $S_f \subseteq [q]^2$ such that $\eff_{i,j}(T) \geq f$ for all $(i,j) \in S_f$, we have $\omega_\ell \leq \log(\rank(T)) / \log(f \cdot \sqrt{|S_f|})$.
\end{theorem}

\begin{proof}
Define $T' = T \otimes T^{\top}$ and $S'_f = S_f \otimes S_f^{\top}$. We can see that $V_x(S'_f) = V_y(S'_f) = V_x(S_f) \cdot V_y(S_f)$. By Lemma~\ref{lem:symmetrize}, this is at most $|S_f|^3 = |S'_f|^{1.5}$. Furthermore, for any $i,i',j,j' \in [q]$ such that $((i,j'), (j,i')) \in S'_f$, we have that $\eff_{(i,j'), (j,i')}(T') = (\eff_{i,j}(T)) \cdot (\eff_{j',i'}(T^{\top})) \geq f^2$. We may thus apply Theorem~\ref{thm:main2} to $T'$ and $S'_f$ to yield the desired result.
\end{proof}

\begin{theorem}[Restatement of Theorem~\ref{thm:mainintro}]\label{thm:main_new}
Suppose $T$ is a $\langle q,q,q_k\rangle$-sized tensor, then $$\omega_\ell \leq \frac{\log(\rank(T))}{\log(\eff(T))}.$$
\end{theorem}
\begin{proof}
Let $N$ be a sufficiently large positive integer. We will partition the set $[q]^N\times [q]^N$ into many regular sets in the following way. For any $(I,J)=((i_1,\cdots,i_N),(j_1,\cdots,j_N))\in ([q]^N)^2$, let $p\in [q]^2\rightarrow \N$ be the counter, such that $p(u,v)$ counts the number of pairs $(i_\ell,j_\ell)$ that equal to $(u,v)$, and let $S_p \subseteq [q]^{2N}$ be the set including all such pairs $(I,J)$ whose counter is $p$. We have that $\{S_p\}_{\{p\}}$ form a partition of $[q]^N\times [q]^N$. We can also see that for every $p$, $S_p$ is a regular set (since its definition does not depend on the order of the $N$ indices). Let $f_p = \Pi_{i,j\in[q]}{\eff}_{i,j}(T)^{p(i,j)}$, every pair $(I,J)\in S_p$ has ${\eff}_{I,J}(T^{\otimes N}) = f_p$. 
By Theorem~\ref{thm:main3}, we have
\begin{align*}
    w_{\ell} \leq \frac{\log(\rank(T^{\otimes N}))}{\log(\sqrt{f_p^2 \cdot |S_p|})}.
\end{align*}
The next step is to choose the best $p$ that maximize $f_p^2\cdot |S_p|$. Note that the number of different $p$ is upper bounded by $N^{q^2}$. And thus
\begin{align*}
    \max_p f_p^2\cdot |S_p| 
    \geq \frac{1}{N^{q^2}} \sum_p f_p^2|S_p|
    = \frac{1}{N^{q^2}} \sum_p \sum_{(I,J)\in S_p} {\eff}_{I,J}(T^{\otimes N})^2
    = \frac{1}{N^{q^2}} {\eff}(T^{\otimes N})^2,
\end{align*}
where the last step is because $\{S_p\}_{\{p\}}$ is the partition of $[q]^N\times [q]^N$ and the definition of ${\eff}(T)$.

Therefore, we have
\begin{align*}
    w_{\ell} \leq \frac{\log(\rank(T^{\otimes N}))}{\log(\sqrt{f_p^2 \cdot |S_p|})} \leq
    \frac{\log(\rank(T^{\otimes N}))}{\log\left(\sqrt{  \frac{1}{N^{q^2}} {\eff}(T^{\otimes N})^2 }\right)} =
    \frac{N\log(\rank(T))}{\log(\sqrt{  \frac{1}{N^{q^2}}}) + N\log({\eff}(T))} \leq \frac{\log(\rank(T))}{\log({\eff}(T))} + o(1),
\end{align*}
and the result follows from taking $N\rightarrow \infty$.
\end{proof}

\section{Solving the $P$-light bulb problem with locality-sensitive hashing} \label{sec:lsh}

\paragraph{General Faster Algorithm.} The general statement of Theorem~\ref{thm:intro2112} which applies to any tensor needs a few definitions. Let $q \geq 2$ be an integer, and $P \in \R^{q \times q}$ be a matrix of nonnegative real numbers whose entries sum to $1$, but whose entries are \emph{not} all equal to $1/q^2$. We say that two vectors $x, y \in [q]^d$ are \emph{jointly sampled according to $P$} if, for each $\ell \in [d]$, the coordinates $x[\ell], y[\ell]$ are sampled independently of all other coordinates, and $(x[\ell], y[\ell]) = (i,j)$ with probability $P[i,j]$ for all $(i,j) \in [q]^2$.

We focus on a generalization of the light bulb problem which our algorithm is naturally able to solve. In the \emph{$P$-light bulb problem}, one is given as input vectors $x_1, \ldots, x_n, y_1, \ldots, y_n \in [q]^d$ which are all independent and uniformly random except for a planted pair which has been jointly sampled according to $P$, and the goal is to find the planted pair. The light bulb problem with correlation $\rho$ is a special case of this problem with $q=2$, $P[0,0] = P[1,1] = (1 + \rho)/4$, $P[0,1] = P[1,0] = (1 - \rho)/4$. It could alternatively be viewed as a special case of this problem for any $q$ which is a power of $2$, along with the appropriately defined $P$.

\subsection{Overview of the proof}
Let $x^*,y^*$ be the correlated pair. Since each bit of $(x^*_i,y^*_i)$ is sampled according to the joint distribution $P$, the number of coordinates $l$ such that $(x^*_l,y^*_l) = (i,j)$ will be proportional to $P[i,j]$ in expectation. In fact, we will assume that the number is equal to its expectation for all $(i,j)$; this happens with decent probability, and will simplify our analysis. The assumption below that $(x^*,y^*)$ falls into the set $V_N$ (defined in Eq.~\eqref{eq:vn} below) captures this property.

Given a tensor $T$ of size $\langle q,q,q \rangle$ along with its $\eff$ matrix, our hope is that $(x^*,y^*)$ falls into a bucket $(i,j)\in [q]^2$ with high $\eff_{i,j}$. However, depending on how the $\eff$ matrix correlates with the $P$ distribution matrix, this may not be the case. To address this, we will choose a pair of stochastic matrices $Q_x,Q_y$ (Def.~\ref{def:effQ_gamma} below) which we use to process vectors after they have been sampled, so that $P$ after this transformation will be correlated with $\eff$.

More precisely, we use $Q_x$ and $Q_y$ to decide which bucket every vector goes into as follows. Take a vector $x$ as example. For every coordinate $l$, if $x_l = i$, we switch $x_l$ to $j$ with probability $Q_x[i,j]$. The final $x$ after this transformation is the bucket we put this vector into. Vectors $y$ are transformed in a similar way, but with the matrix $Q_y$. Note that different buckets may have different numbers of points since the matrices $Q_x,Q_y$ are not necessarily doubly-stochastic. This needs to be taken into account since $(x^*_i,y^*_i)$ wil only be detected it they are put into a bucket where $\eff_{i,j}^2$ is larger than the number of pairs of points. Below we will rescale $\eff$ by $\partial_x$ and $\partial_y$ (Def.~\ref{def:effQ_gamma} below) to ``normalize'' this effect.

Ultimately we need to optimize over choices of $Q_x,Q_y$ to achieve the best running time. The number $\gamma_{Q_x,Q_y}$ (Def.~\ref{def:effQ_gamma} below) indicates the performance of particular matrices $Q_x,Q_y$. The higher this number is, the better the choice of $Q_x,Q_y$.

Ultimately we will find using properties of the Kronecker power that many buckets shares the same property: they have the same chance that of containing the correlated pair $(x^*,y^*)$, and they have the same $\eff$ value. We cluster these buckets into many groups, each specified by a mapping $\tau$ (Def.~\ref{def:tau_partition} below). We find the best cluster $S_{\tau}$ in Lemma~\ref{lem:best_tau} below, and our algorithm will only consider the buckets inside this cluster to find $(x^*,y^*)$. (We will calculate that other clusters give a negligible additional probability of finding $(x^*, y^*)$.) Similar to Theorem~\ref{thm:main2} above, we copy $x^*$ and $y^*$ multiple times to guarantee that there is a constant probability that at least one copy falls into a bucket in that cluster.

Similar to before, we will aim to use Lemma~\ref{lem:combrect2}, and toward this goal, we need to ensure our cluster $S_\tau$ is not too ``skewed''. Similar to Section~\ref{sec:symmetrizing} above, we consider the Kronecker power of $S_{\tau}$ with its transpose $S_{\tau}^{\top}$ to ``symmetrize'' the cluster and avoid this issue. Section~\ref{sec:symmetrizing2} below is devoted to dealing with this issue.

See Algorithm~\ref{alg:alg1} below for the full algorithm description.

\subsection{Preliminaries}
Given any tensor $T$, we now define its $P\text{-}\eff(T)$, a generalization of $\eff(T)$. We start by defining some useful functions.

\begin{definition}[$\effQ$ and $\gamma$]\label{def:effQ_gamma}
    Suppose $T$ is a  $\langle q,q,q_k\rangle$-sized tensor. Given two stochastic matrices $Q_x,Q_y \in \R^{q\times q}$, define $\partial_x(i) := \sum_{j\in[q]} Q_x[i,j]$ for $i\in [q]$, and same for $\partial_y$. 
    we define the Q version of $\eff(T)$ as follows. For $i,j\in [q]$,
\[
\effQ_{i,j}(Q_x,Q_y,T) := \frac{\sum_{k \in [q_k]} T(\X_{i,k} \Y_{j,k} \Z_{i,j})}{\sqrt{\sum_{i', j' \in [q], k,k' \in [q_k]} T(\X_{i',k} \Y_{j',k'} \Z_{i,j})^2\partial_x(i')\partial_y(j')}}.
\]

Given a joint probability matrix $P\in \R^{p\times p}$, we further define the performance of $Q_x,Q_y$ as
\begin{align*}
\gamma_{Q_x,Q_y} := \prod_{i,j\in[q]}\left( \sum_{u,v\in [q]} Q_x[i,u] Q_y[j,v] (\effQ_{u,v}(Q_x,Q_y,T))^2 \right)^{P[i,j]}.
\end{align*}

Let $\gamma$ be the best $\gamma_{Q_x,Q_y}$ over all stochastic matrices $Q_x, Q_y$.
\begin{align*}
    \gamma := \max_{Q_x,Q_y}\gamma_{Q_x,Q_y},
\end{align*}
\end{definition}

Note that $\effQ$ is multiplicative: suppose $T, T'$ both have size $\langle q,q,q_k \rangle$, and $Q_x$, $Q_x'$, $Q_y$, $Q_y'$ are all stochastic matrices, then for $(i,i'), (j,j') \in [q]^2$, we have $\effQ_{(i,i'), (j,j')}(Q_x\otimes Q_x', Q_y\otimes Q_y', T \otimes T') = \effQ_{i,j}(Q_x,Q_y,T) \cdot \effQ_{i',j'}(Q_x',Q_y',T').$

\begin{remark} \label{remark:uniform}
By choosing matrix $Q_x,Q_y$ to be the matrix where each entry is $1/q$, $\effQ(T)$ is the same as $\eff(T)$ and thus $\gamma \geq \eff(T)^2 / q^2$.
\end{remark}

We finally define:
\begin{align*}
\Peff(T):= \sqrt{\gamma \cdot q^2}.
\end{align*}

Here's our main theorem of this section.
\begin{theorem}\label{thm:main_new_2}
Suppose $T$ is a $\langle q,q,q_k\rangle$-sized tensor. Suppose there are $n$ vectors uniformly independently sampled from $[q]^d$, with a planted pair $(x^*,y^*)$ where each bit of them is sampled from a symmetric joint probability matrix $P\in \R^{q\times q}$.

Let $\gamma$ be as defined in Definition~\ref{def:effQ_gamma}, then $(x^*,y^*)$ can be found in $O(n^{\omega_P+o(1)})$ time, where 
\begin{align*}
    \omega_P \leq \frac{\log \rank(T)}{\log (q \gamma^{1/2})}.
\end{align*}
\end{theorem}

Theorem~\ref{thm:intro2112} follows from Theorem~\ref{thm:main_new_2} by finding the optimal $Q$ for the tensor $T_{2112}$ (see Section~\ref{sec:newtensor} for the definition of $T_{2112}$) and the matrix $P$ arising from $\rho$ in the light bulb problem; see the Example~\ref{example:2112} below for more details.

\begin{example} \label{example:2112}
Consider our standard light-bulb problem where the correlated pair has $\rho$-correlation, i.e., 
\begin{align*}
    P_{\rho} = \begin{pmatrix}
    \frac{1+\rho}{4} & \frac{1-\rho}{4}\\
    \frac{1-\rho}{4} & \frac{1+\rho}{4}
    \end{pmatrix}.
\end{align*}
Let $a \in [0,1]$ be some parameter, let both $Q_x$ and $Q_y$ be
\begin{align*}
    Q_x = Q_y = \begin{pmatrix}
    1-a & a\\
    a & 1-a
    \end{pmatrix}.
\end{align*}
Given our $T_{2112}$ tensor, by definition of $\gamma_{Q_x,Q_y}$ (Definition~\ref{def:effQ_gamma}), we can calculate
\begin{align*}
    \gamma_{Q_x,Q_y} = \left( 2\cdot \big((1-a)^2+a^2\big) + 1\cdot \big(2a(1-a)\big) \right)^{\frac{1+\rho}{2}}\left( 2\cdot \big(2a(1-a)\big) + 1\cdot \big((1-a)^2+a^2\big) \right)^{\frac{1-\rho}{2}}.
\end{align*}
The optimal $a$ is given as $a = \max\{0, (1 - \sqrt{3\rho})/2\}$.
Write $\omega_{\rho}$ to be the exponent $\omega_P$ in Theorem~\ref{thm:main_new_2} given our $P = P_{\rho}$. Then we get
\begin{align*}
    w_{\rho} = \begin{cases}
    \frac{2\log 5}{\log\left(6(1-\rho)^{-\rho/2}(1+\rho)^{\rho/2}(1-\rho^2)^{1/2}\right) } & \textrm{,when~}\rho < 1/3\\
    \frac{4\log 5}{(5+\rho)\log 2} & \textrm{,when~}1/3\leq \rho \leq 1.
    \end{cases}
\end{align*}
\end{example}

In the remainder of this section, we prove Theorem~\ref{thm:main_new_2}.

\subsection{Preparation before symmetrization}
We will show that for every stochastic $Q_x,Q_y$, our algorithm gives an exponent $\omega_{P} \leq 2\frac{\log \rank(T)}{\log(\gamma_{Q_x,Q_y})}$ and therefore, the theorem follows by choosing the best $Q_x$ and $Q_y$. In the following, we assume $Q_x, Q_y$ are fixed stochastic matrices. 

Fix a joint probability matrix $P$. Let $N$ be a sufficiently large positive integer so that $P[i,j]N$ are all integers.\footnote{We only need this property when constructing the set $V_N$. There is a negligible change in our algorithm if we round $P[i,j]\cdot N$ to be the integer closest to $P[i,j]\cdot N$ in the construction of $V_N$ for large $N$.} Define the set $V_N$ as all pairs of $(x,y)$ where there are exactly $P[i,j]\cdot N$ number of coordinates $l$ such that $(x[l],y[l])=(i,j)$,
\begin{align}\label{eq:vn}
V_N := \{ (x,y)\in [q]^{N} : \forall i,j\in [q],~|\{l\mid (x[l],y[l]) = (i,j)\}| = P[i,j]\cdot N\}.
\end{align}
If the planted pair $(x^*,y^*)$ is drawn from the joint distribution $P$, then there's a descent chance that $(x^*,y^*)\in V_N$.

\begin{definition}[Distribution $\mathcal{D}^{x}$ and $\mathcal{D}^{y}$] \label{def:Dxy}
Let $i \in [q]$ and $Q_x \in \R^{q \times q}$ be a stochastic matrix, we define the distribution $\mathcal{D}_{Q_x}^{i} \in \R^{q}$ as
\[
\mathcal{D}_{Q_x}^{i}(j) := Q_x[i,j], ~\forall j\in [q].
\]
In another words, the distribution $\mathcal{D}_{Q_x}^{i}$ is generated by transforming $i$ to $j$ with probability $Q_x[i,j]$. 

Let $x\in [q]^{N}$ be any vector, we define the distribution $\mathcal{D}_{Q_x^{\otimes N}}^{x}  \in \R^{q^N}$ as 
\[
\mathcal{D}_{Q_x^{\otimes N}}^{x} := \otimes_{l=1}^N \mathcal{D}_{Q_x}^{x[l]}.
\]
In another words, the distribution $\mathcal{D}_{Q_x^{\otimes N}}^{x}$ is generated by transforming  independently each entry $x[\ell]$ to $x'[\ell]$ with probability $Q_x[x[\ell],x'[\ell]]$.

Similarly, we define
\begin{align*}
    \mathcal{D}^{y}_{Q_y^{\otimes N}} := \otimes_{l=1}^N \mathcal{D}^{y[l]}_{Q_y} \in \R^{q^N}.
\end{align*}

For simplicity, if $Q_x$ and $Q_y$ are clear from the context, we write $\mathcal{D}^{x}_{Q_x^{\otimes N}}$ as $\mathcal{D}^{x}$ and $\mathcal{D}^{y}_{Q_y^{\otimes N}}$ as  $\mathcal{D}^{y}$.

We define $\mathcal{D}^{x,y} := \mathcal{D}^x \otimes \mathcal{D}^y$ to be their joint distribution. In particular, for $x', y' \in [q]^n$, we will write $\mathcal{D}^{x,y}(x',y') = \mathcal{D}^{x}(x') \cdot \mathcal{D}^{y}(y')$ to denote the probability that $\mathcal{D}^{x,y}$ outputs $(x',y')$.
\end{definition}

Next, similar to the proof of Theorem~\ref{thm:main_new}, we will partition the entire space $[q]^N \times [q]^N$ into several regular sets.

\begin{definition}[$\tau$-partition]\label{def:tau_partition}
Call a mapping $\tau :[q]^4\rightarrow \{0,1,\cdots,N\}$ \emph{valid} if for all $i,j\in[q]$, we have $\sum_{u,v\in[q]} \tau(i,j,u,v) = N \cdot P[i,j]$. Fix a pair $(x^*,y^*)\in V_N$.
Every pair $(x,y)\in [q]^{N} \times [q]^N$ \emph{corresponds to} one valid mapping $\tau$ defined by $\tau(i,j,u,v) = |\{ l : (x^*[l],y^*[l]) = (i,j)~\textrm{and}~(x[l],y[l])=(u,v)\}|$. For a valid mapping $\tau$, let $S_{\tau}^{x^*,y^*}\subseteq[q]^{2N}$ be the set of all pairs $(x,y)$ that correspond to $\tau$. When $x^*,y^*$ is clear from the context, we simply write $S_{\tau}$ instead of $S_{\tau}^{x^*,y^*}$.
\end{definition}

We next make some key observations about valid mappings $\tau$.

\begin{fact}\label{fac:tau_partition}
Fix $(x^*,y^*)\in V_N$, then
\begin{enumerate}
    \item $\{S_{\tau}\}_{\{\text{valid }\tau\}}$ is a partition of $[q]^N \times [q]^N$;\label{fac:part1}
    \item $S_{\tau}$ is regular for all valid $\tau$;\label{fac:part2}
    \item Every $(x,y)\in S_{\tau}$ has the same $\mathcal{D}^{x^*,y^*}(x,y)$, since $\mathcal{D}^{x^*,y^*}(x,y) = \prod_{i,j,u,v}(Q_x[i,u]Q_y[j,v])^{\tau(i,j,u,v)}$ only depends on $\tau$. For simplicity, we denote them as $\mathcal{D}_{\tau}$;
    \label{fac:part3}
    \item Every $(x,y)\in S_{\tau}$ has the same $\effQ_{x,y}(T^{\otimes N})$, since 
    \[
    \effQ_{x,y}(Q_x^{\otimes N},Q_y^{\otimes N},T^{\otimes N}) = \prod_{i,j,u,v}\effQ_{u,v}(Q_x,Q_y,T)^{\tau(i,j,u,v)}
    \]
    only depends on $\tau$. For simplicity, we denote them as $\effQ_{\tau}$. \label{fac:part4}
\end{enumerate}
\end{fact}
\begin{proof}
1. This is because each pair $(x,y)$ corresponds to exactly one valid $\tau$.

2. For $i,j \in [q]$, let $Z_{i,j}:=\{l\mid x^*[l]=i~\textrm{and}~y^*[l]=j\}\subseteq[N]$, and let $(S_{\tau})_{i,j} \subseteq [q]^{Z_{i,j}}$ be the set of pairs $(x,y) \in [q]^{Z_{i,j}}$ for which, for all $(u,v) \in [q]^2$, we have $|\{ \ell \in Z_{i,j} \mid (x[\ell],y[\ell])=(u,v) \}| = \tau(i,j,u,v)$. We can see that $(S_{\tau})_{i,j}$ is regular since it is defined independently of the order of the indices. Thus, $S_{\tau}=\otimes_{i,j} (S_{\tau})_{i,j}$, which is a Kronecker product of regular sets $(S_{\tau})_{i,j}$, is also regular.

3 and 4. Proved in the statement.
\end{proof}

\begin{lemma}\label{lem:best_tau}
Let $N$ be a sufficient large integer.
Let $(x^*,y^*)\in V_N$ be any pair.
Then there exists a valid mapping $\tau$ such that
\[
\mathcal{D}_{\tau} \cdot |S_{\tau}| \cdot {\effQ}_{\tau}^2 \geq \frac{1}{(N+1)^{q^4}} \gamma_{Q_x,Q_y}^N,
\]
where $\gamma_{Q_x,Q_y}$ is defined in Definition~\ref{def:effQ_gamma}.
\end{lemma}
\begin{proof}
\begin{align*}
    & ~ \max_{\tau} \mathcal{D}_{\tau} \cdot |S_{\tau}| \cdot {\effQ}_{\tau}^2 \\
    \geq & ~ \frac{1}{(N+1)^{q^4}} \sum_{\tau} \mathcal{D}_{\tau} \cdot |S_{\tau}| \cdot {\effQ}_{\tau}^2\\
    = & ~ \frac{1}{(N+1)^{q^4}}\sum_{\tau} \E_{(x,y)\sim \mathcal{D}^{x^*,y^*}}[{\effQ}_{x,y}^2(T^{\otimes N})\cdot \big[(x,y)\in S_{\tau}\big]_{\mathbf{1}} ]\\
    = & ~ \frac{1}{(N+1)^{q^4}}\E_{(x,y)\sim \mathcal{D}^{x^*,y^*}}[{\effQ}^2_{x,y}(T^{\otimes N})],
\end{align*}
where the first step follows because there are at most $(N+1)^{q^4}$ different valid $\tau$, the second step follows from part~\ref{fac:part3} and part~\ref{fac:part4} of Fact~\ref{fac:tau_partition}, the third step follows from part~\ref{fac:part1} of Fact~\ref{fac:tau_partition}. We conclude by computing that
\begin{align*}
& ~ \E_{(x,y)\sim \mathcal{D}_{x^*,y^*}}[{\effQ}^2_{x,y}(T^{\otimes N})]\\
= & ~ \E_{ \substack{(x_1,y_1)\sim \mathcal{D}^{x_1^*,y_1^*}\\~\vdots~\\(x_N,y_N)\sim\mathcal{D}^{x_N^*,y_N^*}}}[ \Pi_{i=1}^N {\effQ}^2_{x_i,y_i}(T)]\\
= & ~ \prod_{i=1}^N\E_{ (x_i,y_i)\sim\mathcal{D}^{x_i^*,y_i^*}} [{\effQ}^2_{x_i,y_i}(T)] \\
= & ~ \prod_{i,j\in[q]}\left(\sum_{u,v\in[q]} Q_x[i,u]Q_y[j,v]{\effQ}_{u,v}^2(T)\right)^{P_{i,j}N}
\\= & \gamma_{Q_x,Q_y}.\qedhere
\end{align*}

\end{proof}

We also give this lemma for later use.
\begin{lemma}\label{lem:mapto1}
Suppose $q$ is even, and $P \in \R^{q\times q}$ is a joint probability matrix. There exist two mappings $h,g: [q]\rightarrow \{-1,1\}$ and a constant $\rho>0$ such that: 
\begin{itemize}
    \item If $b_0, b_1$ are sampled independently from $[q]$, and at least one of them is sampled uniformly, then $\E[g(b_0) \cdot h(b_1)] = 0$, but
    \item If $b_0, b_1$ are sampled from $[q]$ according to $P$ (so $(b_0, b_1) = (i,j)$ with probability $P[i,j]$), then $\E [g(b_0) \cdot h(b_1)] = \rho$.
\end{itemize}

\end{lemma}
\begin{proof}
We construct $g,h$ in a greedy fashion. Pick any row of $P$ that is not uniform. Such a row exists since we assume $P$ is not the uniform matrix. Fix $g$ to map the column indices of the $q/2$ largest entries in that row to $1$, and the others to $-1$. Let $v:=P\cdot g\in \R^{q}$. Fix $h$ to map the indices of the $q/2$ largest entries of $v$ to $1$ and the others to $-1$.

If $b_0$ or $b_1$ is sampled uniformly from $[q]$, then $g(b_0)$ or $h(b_1)$, respectively, is uniformly chosen from $\{-1,1\}$ since $h$ and $g$ each map half of $[q]$ to $1$ and the other half to $-1$. Hence, in this case, if $b_0$ and $b_1$ are sampled independently, then $\E[g(b_0) \cdot h(b_1)] = 0$.

Meanwhile, by our construction of $h$,
\[
\E_{(b_0, b_1) \sim P}[g(b_0)\cdot h(b_1)] = \langle v, h\rangle > 0,
\]
which is strictly larger than $0$ because $v$ is non-zero. We may thus pick $\rho = \langle v, h\rangle$.
\end{proof}

\begin{algorithm}[!t]\caption{}\label{alg:alg1} 
\begin{algorithmic}[1]
\State \textbf{Input}
\State Let $T$ be a $\langle q,q,q_k \rangle$-sized tensor.
\State Let $x_1,\cdots,x_n,y_1,\cdots,y_n \in [q]^d$ be $2n$ vectors with one planted pair $(x^*,y^*)$.
\State Let $P\in \R^{q \times q}$ be joint-probability symmetric matrix for planted pair, that every bit $(x^*[l],y^*[l])$ is sampled according to $P$.
\State \textbf{Algorithm}
\State Let $\gamma$ and its corresponding stochastic matrices $Q_x,Q_y\in \R^{q\times q}$ be from Definition~\ref{def:effQ_gamma}.
\State Let $N$ be such that $q^{2N}\gamma^{N} = 20n\cdot (N+1)^{q^4}$
\State $T'\leftarrow (T\otimes T^{\top})^{\otimes N}$
\For{$g$ a power of $2$ from $1,2,4\cdots,\max_{i,j}\eff_{i,j}(T')$} \label{line:for_start}
    \State Prepare $q^{2N}$ sets indexed by vector in $[q]^{2N}$ for both $x$ and $y$, $X_1,\cdots,X_{q^{2N}}$, $Y_1,\cdots,Y_{q^{2N}}$.
    \State $c\leftarrow q^{2N}\cdot g/n$ \label{line:c}
    \State for each $x_i$, independently generate $c$ indices $i_1,\cdots,i_{c}$ from 
     $\mathcal{D}^{x_i}_{Q_x^{\otimes N}\otimes Q_y^{\otimes N}}$ and add $\{i\}$ to every $X_{i_j}$.
     \label{line:make_copy_1}
    \State for each $y_i$, independently generate $c$ indices $i_1,\cdots,i_{c}$ from
     $\mathcal{D}^{y_i}_{Q_y^{\otimes N}\otimes Q_x^{\otimes N}}$ and add $\{i\}$ to every $Y_{i_j}$. \label{line:make_copy_2}
    \State Find correlated pair under such groupings. \Comment{ See details in Section~\ref{sec:step4}} 
\EndFor \label{line:for_end}
\end{algorithmic}
\end{algorithm}

\subsection{Proof of Theorem~\ref{thm:main_new_2}} \label{sec:proof}
In Lemma~\ref{lem:best_tau} from the previous section, we found the best mapping $\tau$, and we aimed to use $S_{\tau}$ to detect the correlated pair. However, although $S_{\tau}$ is a regular set, it still can be ``skewed'' (in the sense of Section~\ref{subsec:skew}). In this section, at a high level, we are going to symmetrize $S_{\tau}$ using a Kronecker product with its transpose $S_{\tau}^{\top}$ to avoid skew. We will prove Theorem~\ref{thm:main_new_2} in four steps.

\subsubsection{Step 1. General start}

We may assume $\gamma > 1/q$, since otherwise, the bound on $\omega_P$ we claim in Theorem~\ref{thm:main_new_2} is worse than the trivial exponent of $2$. We first fix $N$ such that 
\begin{align}\label{eq:q2n}
    q^{2N} = 20n\left(\frac{1}{\gamma}\right)^{N}(N+1)^{q^4}.
\end{align}
Note that $N = \Theta_q(\log n)$.

For simplicity, let us assume that the input vectors are long enough; as discussed in the introduction and early proofs, that one can use the `compressed matrices' method introduced in~\cite[Section 4.2]{karppa2018faster} to `expand' lower-dimensional vectors without losing too much correlations on correlated pair. 

Let the planted pair be $(x^*,y^*)$. We abuse notation here and also write $x^*,y^*\in [q]^{2N}$ to denote the first $2N$ coordinates of the planted pair. We write $x^* = x_1^* \circ x_2^*$, where $x^*_1,x^*_2\in [q]^N$, and write $y^* = y_1^* \circ y_2^*$ in the same way. (Here $\circ$ denotes vector concatenation.)

We will use Algorithm~\ref{alg:alg1} to solve the problem. As suggested by line~\ref{line:for_start} to line~\ref{line:for_end} of the algorithm, our goal here is to copy every vector $c$ times and partition them into $q^{2N}$ groups, with each group containing roughly $g$ vectors (so $nc/g\approx q^{2N}$). Our algorithm enumerates over $g$ from $1$ to $(\max_{i,j} \eff_{i,j})^{2N}$ by doubling each time. We are going to prove that we will successfully find the correlated pair for one of these choices of $g$.

We set aside the first $2N$ entries of each input vector which we will use to decide the grouping. We will later use fresh entries from the input vectors in later parts of the algorithm (when we perform matrix multiplication) so that there is no correlation between the independent random vectors which are placed in the same group.

With probability $1/(q^N)^{o(1)}$, we have both pairs $(x^*_1,y^*_1),(x^*_2,y^*_2) \in V_N$. We will assume this happens in the later analysis, since it only cost a $1/(q^N)^{o(1)}$ overhead on the running time by repeating the algorithm using fresh bits. 

\subsubsection{Step 2. Symmetrizing $S_{\tau}$}\label{sec:symmetrizing2}

Recall that $Q_x,Q_y \in \R^{q\times q}$ are the given stochastic matrices. We apply $Q_x$ to $x_1^*$ and $Q_y$ to $y_1^*$, and let $\tau_1$ be the best $\tau$ chosen from Lemma~\ref{lem:best_tau} with respect to $x_1^*,y_1^*,Q_x,Q_y$. For ease of presentation, we will still use notation $S_{\tau}$, $\mathcal{D}_{\tau}$, $\effQ_{\tau}$ in the Lemma~\ref{lem:best_tau} and note that they are with respect to $\tau_1$ and $x_1^*$, $y_1^*$, $Q_x$, $Q_y$, i.e., we have
\begin{align}\label{eq:StauDtaueffQtau}
   \notag S_{\tau} = & ~  S_{\tau_1}^{x_1^*,y_1^*}, ~~~\textrm{and}\\ 
    \forall (x,y)\in S_{\tau},~
    \mathcal{D}_{\tau} = & ~ \mathcal{D}_{Q_x^{\otimes N},Q_y^{\otimes N}}^{x_1^*,y_1^*}(x,y),\\ \notag
    \forall (x,y)\in S_{\tau}, ~\effQ_{\tau} = & ~ \effQ_{x,y}(Q_x^{\otimes N},Q_y^{\otimes N},T^{\otimes N}).
\end{align}

Now, consider symmetrizing $S_{\tau}$ as follows. Let $\tau_2$ be the (``transposed'') mapping such that, for all $i,j,u,v \in [q]$, we have $\tau_2(i,j,u,v) := \tau_1(j,i,v,u)$. Since $P$ is a symmetric matrix, we know that $\tau_2$ is also \emph{valid}. Define $S := S_{\tau_1}^{x_1^*,y_1^*}\otimes S_{\tau_2}^{x_2^*,y_2^*} \subseteq [q]^{2N}$, and we will show in Claim~\ref{cla:transpose}, $S = S_{\tau}\otimes S_{\tau}^{\top}$. For the purpose of symmetricity, we will further apply $Q_y$ to $x_2$ and $Q_x$ to $y_2$. We are able to do so because $P$ is symmetric and thus, $(y_2^*,x_2^*)$ is also in $V_N$.

We will prove that, with a careful choice of $c$ (the number of copies we make of each vector), there will be a decent probability that a copy of the planted pair falls into some bucket in $S$. To do this, we first need to prove $S$ is not skewed. The following are some useful fact about $S$.

\begin{claim}\label{cla:transpose}
$(S_{\tau_1}^{x_1^*,y_1^*})^{\top} = S_{\tau_2}^{y_1^*,x_1^*}$.
\end{claim}
\begin{proof}
\begin{align*}
& ~ (S_{\tau_1}^{x_1^*,y_1^*})^{\top} \\
= & ~ \{ (x,y)\in [q]^{2N} : \forall i,j,u,v\in[q]~\left| \{ \ell\in [N] : x_1^*[\ell]=i~\textrm{and}~y_1^*[\ell]=j~\textrm{and}~x[\ell]=v~\textrm{and}~y[\ell]=u\}\right| = \tau_1(i,j,u,v) \} \\
= & ~ \{ (x,y)\in [q]^{2N} : \forall i,j,u,v\in[q]~\left| \{ \ell\in [N] : x_1^*[\ell]=i~\textrm{and}~y_1^*[\ell]=j~\textrm{and}~x[\ell]=v~\textrm{and}~y[\ell]=u\}\right| = \tau_2(j,i,v,u) \} \\
= & ~ \{ (x,y)\in [q]^{2N} : \forall i,j,u,v\in[q]~\left| \{ \ell\in [N] : x_1^*[\ell]=j~\textrm{and}~y_1^*[\ell]=i~\textrm{and}~x[\ell]=u~\textrm{and}~y[\ell]=v\}\right| = \tau_2(i,j,u,v) \} \\
= & ~ S_{\tau_2}^{y_1^*,x_1^*},
\end{align*}
where the first step is by definition, the second step replaces $\tau_1$ with $\tau_2$, and the third step is by switching variables.
\end{proof}

\begin{definition}
For a ground set $U$, we say two sets $S_1, S_2 \subseteq U^2$ are \emph{isomorphic} if they are equal up to permuting first and second coordinates, i.e., there are permutations $\pi_1, \pi_2 : U \to U$ such that, for all $(a,b) \in U^2$, $(a,b) \in S_1$ if and only if $(\pi_1(a), \pi_2(b)) \in S_2$.
\end{definition}

\begin{claim}\label{cla:self_tensor_property}
1. $V_x(S),V_y(S)\leq |S|^{1.5}$;~~~~~~~~~~~ ($V_x(S)$, $V_y(S)$ is defined in Definition~\ref{def:v_x_v_y})

2. $|S| = |S_{\tau}|^2$;

3. $\forall (x,y)\in S$, $\mathcal{D}^{x^*}_{Q_x^{\otimes N}\otimes Q_y^{\otimes N}}(x) = \mathcal{D}^{y^*}_{Q_y^{\otimes N}\otimes Q_x^{\otimes N}}(y) = \mathcal{D}_{\tau}$;

4. Let $\wt{T}:=T^{\otimes N}$. For all $(x,y)\in S$, we have ${\effQ}_{x,y}(Q_x^{\otimes N}\otimes Q_y^{\otimes N}, Q_y^{\otimes N}\otimes Q_x^{\otimes N}, \wt{T}\otimes \wt{T}^{\top}) = \effQ_{\tau}^2$.
\end{claim}
\begin{proof}
\textbf{Part 1 and 2.}

For any valid $\tau$ and $(x_1,y_1),(x_2,y_2)\in V_N$, the sets $S^{x_1,y_1}_{\tau}$ and $S^{x_2,y_2}_{\tau}$ are isomorphic, since the definition of $S_{\tau}$ does not depend on the order of the $N$ indices. Therefore, because both $(y_1^*,x_1^*)$ and $(x_2^*,y_2^*)$ are in $V_N$, we know that $S^{y_1^*,x_1^*}_{\tau_2}$ is isomorphic to $S^{x_2^*,y_2^*}_{\tau_2}$. By Claim~\ref{cla:transpose}, 
$(S_{\tau_2}^{y_1^*,x_1^*})^{\top} = S_{\tau_1}^{x_1^*,y_1^*}$, so
\begin{align*}
    |S| = |S_{\tau_1}^{x_1^*,y_1^*}|\cdot |S_{\tau_2}^{x_2^*,y_2^*}| = |S_{\tau_1}^{x_1^*,y_1^*}|\cdot |S_{\tau_2}^{y_1^*,x_1^*}| = |S_{\tau_1}^{x_1^*,y_1^*}|^2,
\end{align*}
and
\begin{align*}
    V_x(S) = V_x(S_{\tau_1}^{x_1^*,y_1^*}) \cdot V_x(S_{\tau_2}^{x_2^*,y_2^*}) = V_x(S_{\tau_1}^{x_1^*,y_1^*})\cdot V_y(S_{\tau_1}^{x_1^*,y_1^*}) \leq |S_{\tau_1}^{x_1^*,y_1^*}|^3 = |S|^{1.5},
\end{align*}
where the third step follows because $S_{\tau_1}^{x_1^*,y_1^*}$ is regular (by part~\ref{fac:part2} of Fact~\ref{fac:tau_partition}) and by Lemma~\ref{lem:symmetrize}.

$V_y(S)\leq |S|^{1.5}$ follows for the same reason.

\textbf{Part 3.}
For all $(x,y) \in S$, 
\begin{align*}
    \mathcal{D}^{x^*}_{Q_x^{\otimes N}\otimes Q_y^{\otimes N}}(x) = 
    \prod_{i,j,u,v}(Q_x)_{i,u}^{\tau_1(i,j,u,v)}
    \prod_{i,j,u,v}(Q_y)_{i,u}^{\tau_2(i,j,u,v)} 
    = 
    \prod_{i,j,u,v}(Q_x)_{i,u}^{\tau_1(i,j,u,v)}
    \prod_{i,j,u,v}(Q_y)_{j,v}^{\tau_1(i,j,u,v)}
    =
    \mathcal{D}_{\tau},
\end{align*}
where the second step is by the symmetry of $\tau_1$ and $\tau_2$, and the last step is by the definition of $\mathcal{D}_{\tau}$ (part~\ref{fac:part3} of Fact~\ref{fac:tau_partition}). Similarly,

\begin{align*}
    \mathcal{D}^{y^*}_{Q_y^{\otimes N}\otimes Q_x^{\otimes N}}(y) 
    = \prod_{i,j,u,v}(Q_y)_{j,v}^{\tau_1(i,j,u,v)}
    \prod_{i,j,u,v}(Q_x)_{j,v}^{\tau_2(i,j,u,v)}
    = \prod_{i,j,u,v}(Q_y)_{j,v}^{\tau_1(i,j,u,v)}
    \prod_{i,j,u,v}(Q_x)_{i,u}^{\tau_1(i,j,u,v)}
    = \mathcal{D}_{\tau}.
\end{align*}

\textbf{Part 4.} For all $x,y\in S$, we write $x = x_1\circ x_2$ and $y = y_1\circ y_2$ to partition them into two halves. Then, 
\[
{\effQ}_{x,y}(Q_x^{\otimes N}\otimes Q_y^{\otimes N}, Q_y^{\otimes N}\otimes Q_x^{\otimes N}, \wt{T}\otimes \wt{T}^{\top}) 
= {\effQ}_{x_1,y_1}(Q_x^{\otimes N}, Q_y^{\otimes N}, \wt{T})
\cdot {\effQ}_{x_2,y_2}(Q_y^{\otimes N}, Q_x^{\otimes N}, \wt{T}^{\top}).
\]

We can calculate that
\begin{align*}
{\effQ}_{x_2,y_2}(Q_y^{\otimes N}, Q_x^{\otimes N}, \wt{T}^{\top})
= & ~ \prod_{i,j,u,v}{\effQ}_{u,v}(Q_y,Q_x,T^{\top})^{\tau_2(i,j,u,v)}
= \prod_{i,j,u,v}{\effQ}_{v,u}(Q_x,Q_y,T)^{\tau_2(i,j,u,v)}\\
= & ~ \prod_{i,j,u,v}{\effQ}_{u,v}(Q_x,Q_y,T)^{\tau_1(j,i,u,v)}
= \prod_{u,v}{\effQ}_{u,v}(Q_x,Q_y,T)^{\sum_{i,j}\tau_1(j,i,u,v)}\\
= & ~ {\effQ}_{x_1,y_1}(Q_x,Q_y,T) = {\effQ}_{\tau},
\end{align*}
where the second step is because $\effQ_{u,v}(Q_y,Q_x,T^{\top}) = \effQ_{v,u}(Q_x,Q_y,T)$, the third step is because $\tau_2(i,j,u,v) = \tau_1(j,i,v,u)$ and switching $u$ and $v$, and the last step is by definition of $\eff_{\tau}$ (part~\ref{fac:part4} of Fact~\ref{fac:tau_partition}).

Therefore, the statement holds.
\end{proof}

\subsubsection{Step 3. Detect $(x^*,y^*)$ using $S$}
To simplify notation, we write the distribution $\mathcal{D}^{x}_{Q_x^{\otimes N}\otimes Q_y^{\otimes N}}$ as $\mathcal{D}^x$, omitting the matrix $Q_x^{\otimes N}\otimes Q_y^{\otimes N}$ acting on $x$. Similarly, we write $\mathcal{D}^{y^*}_{Q_y^{\otimes N}\otimes Q_x^{\otimes N}}$ as $\mathcal{D}^{y}$.

We let $S_x:= \{x \mid \mathcal{D}^{x^*}(x) = \mathcal{D}_{\tau}\}$ and $S_y:= \{y \mid \mathcal{D}^{y^*}(y) = \mathcal{D}_{\tau}\}$. By Claim~\ref{cla:self_tensor_property}, $S \subseteq S_x \otimes S_y$, and also $|S_x|=|S_y|$ follows by symmetry.

In line~\ref{line:make_copy_1},~\ref{line:make_copy_2} of our Algorithm~\ref{alg:alg1}, for each $x_i \in [q]^{2N}$ (also $y_i$), we make $c$ copies $i_1, \ldots, i_c$ independently drawn from $\mathcal{D}^{x_i}$, and put $x_i$ into the buckets $X_{i_1}, \ldots, X_{i_c}$. (If two copies turned out to be identical, we only put $x_i$ into that bucket once.) Let $g^*$ be the least power of $2$ that is larger than $\frac{n}{q^{2N} \cdot |S_{\tau}|\cdot D_{\tau}}$ (this number is roughly $\effQ_{\tau}^2$; see Eq.~\eqref{eq:g*}). Since our algorithm iterates over all $g$ which are powers of $2$, there's an iteration where $g = g^*$. In the remaining analysis, we focus on this case. Note that
\begin{align}\label{eq:l}
c = q^{2N}g^*/n \geq \frac{1}{|S_{\tau}|\cdot \mathcal{D}_{\tau}} \geq 1.
\end{align}

\begin{claim}\label{cla:x'_y'_in S}
Suppose $c\geq \frac{1}{|S_{\tau}|\cdot \mathcal{D}_{\tau}}$, then with $\geq 1/4$ probability, there's one copy $x'$ of $x^*$ and one copy $y'$ of $y^*$ such that $(x',y') \in S$.
\end{claim}
\begin{proof}
The copies of the planted vector $x^*$ are drawn independently according to $\mathcal{D}^{x^*}$. Hence, the number of copies which fall into $S_x$ follows a binomial distribution with mean $c\cdot \mathcal{D}_{\tau} \cdot |S_x| \geq 1$. It follows that, with constant probability, there are at least $c\cdot \mathcal{D}_{\tau} \cdot |S_x|$ many copies of $x^*$ which fall into $S_x$. For the same reason, there is a constant probability that $c\cdot \mathcal{D}_{\tau} \cdot |S_y|$ many copies of $y^*$ fall into $S_y$.

Furthermore, since for all $x\in S_x,y\in S_y$, $\mathcal{D}^{x^*}(x) =\mathcal{D}^{y^*}(y) = \mathcal{D}_{\tau}$, it follows that: conditioned on some particular copies of $x^*$ falling into $S_x$, those copies will be independently uniformly random elements of $S_x$ (and similarly for $S_y$). Also, since $|S_x|=|S_y|$ and $c \cdot \mathcal{D}_{\tau}\cdot |S_x|\geq |S_x| / \sqrt{|S|}$ (by the choice of our $c$ and $|S|=|S_{\tau}|^2$) and $V_x(S),V_y(S)\leq |S|^{1.5}$ (Claim~\ref{cla:self_tensor_property}), it follows by Lemma~\ref{lem:combrect2} that with probability at least $1/4$, there is a pair of one copy of $x^*$ and one copy of $y^*$ which falls into $S$.
\end{proof}

\begin{lemma}\label{lem:bound_x_i_y_y}
Given $i,j \in [q]^{2N}$, let $|X_i|$ and $|Y_j|$ be the number of input vectors that placed a copy into $X_i$ and $Y_j$, respectively. Then, over the randomness of the first $2N$ coordinates of the input vectors and the process of making random copies, 
\begin{align*}
    \E[|X_i|] & ~ = g^*\cdot \Pi_{l=1}^{N} \partial_x(i_l) \cdot \Pi_{l=N+1}^{2N} \partial_y(i_l), ~~~ and\\
    \E[|Y_i|] & ~ = g^*\cdot \Pi_{l=1}^{N} \partial_y(j_l) \cdot \Pi_{l=N+1}^{2N} \partial_x(j_l).
\end{align*}
\end{lemma}
\begin{proof}
For every input vector $x$ other than the planted pair, $x$ is uniformly sampled from $[q]^{2N}$. For a copy $x'$ of $x$ drawn from the distribution $\mathcal{D}^{x}_{Q_x^{\otimes N}\otimes Q_y^{\otimes N}}$, we have
\begin{align*}
    \Pr[x' \in X_i] = \frac{1}{q^{2N}}\Pi_{l=1}^{N} \partial_x(i_l) \cdot \Pi_{l=N+1}^{2N} \partial_y(i_l).
\end{align*}

By linearity of expectation over all input vectors and all copies, we have
\begin{align*}
    \E[|X_i|] = nc\cdot \Pr[x' \in X_i].
\end{align*}

Similarly, on the $y$ side, we have
\begin{align*}
    \E[|Y_j|] = \frac{nc}{q^{2N}}\Pi_{l=1}^{N} \partial_y(j_l) \cdot \Pi_{l=N+1}^{2N} \partial_x(j_l).
\end{align*}

Thus, the claim follows since $g^*=\frac{nc}{q^{2N}}$.
\end{proof}

The following gives an upper bound on $g^*$,
\begin{align}\label{eq:g*}
    g^* \leq  \frac{2n}{q^{2N}|S_{\tau}|\mathcal{D}_{\tau}} = \frac{n\gamma^N}{10(N+1)^{q^4}\cdot |S_{\tau}|\mathcal{D}_{\tau}n}\leq \frac{N^{q^4}\mathcal{D}_{\tau}|S_{\tau}|{\effQ}_{\tau}^2}{10(N+1)^{q^4}\cdot |S_{\tau}|\mathcal{D}_{\tau}} = {\effQ}_{\tau}^2/10,
\end{align}
where the second step is by our choice of $N$ from Equation~\eqref{eq:q2n} and $c$  from Equation \eqref{eq:l}, and the third step is by Lemma~\ref{lem:best_tau}.

\subsubsection{Step 4. Matrix multiplication}\label{sec:step4}
Our vectors currently come from $[q]^N$, but we would like to map them to vectors in $\{-1,1\}^N$ so that the independent uniformly random vectors are still independent uniformly random, and the planted pair is correlated.
If $q$ is even, we use the mappings $g$, $h$ from Lemma~\ref{lem:mapto1}. If $q$ is odd, we first map each bit of vectors from $[q]$ to $[2q]$ by adding a uniform bit in $\{0,1\}$, so that the planted pair still has non-zero correlation, then we use the mappings $g$, $h$.

As we discussed earlier, we use fresh bits (different from the ones used in the bucketing process above) for each matrix multiplication. Sample $q_k^{2N}$ coordinates, and apply the mapping $g$ to $x$, and $h$ to $y$, bit-wise. Here we abuse notation and still write $x_i,y_i\in \{-1,1\}^{q_k^{2N}}$ to denote the mapped input vectors $x_i,y_i$. The result is that the mapped vectors $x_i$ and $y_i$ are independently uniformly chosen from $\{-1,1\}^{q_k^{2N}}$, except the correlated pair $x^*,y^*$ has 
\[
\langle x^*,y^*\rangle = \Omega(q_k^{2N}).
\]

For each $i \in [q]^{2N}$, create vectors $a_i, b_i \in \R^{m}$ given by $a_i = \sum_{j \in X_i} x_j$ and $b_i = \sum_{j \in Y_i} y_j$. Let $s_a,s_b\in \{-1,1\}^{q^{2N}}$ be random vectors whose entries are i.i.d. uniformly sampled from $\{-1,1\}$.
Form the matrices $A, B \in \R^{q^{2N} \times m}$ whose rows are $s_a[1]\cdot a_1, \ldots, s_a[q^{2N}]\cdot a_{q^{2N}}$ and $s_b[1]\cdot b_1, \ldots, s_b[q^{2N}]\cdot b_{q^{2N}}$, respectively. 

We now apply the tensor $T'$ to the matrices $A$ and $B^{\top}$, resulting in the matrix $C \in \R^{q^{2N} \times q^{2N}}$. 
By Claim~\ref{cla:x'_y'_in S}, with $\geq 1/4$ probability, one copy of $x^*$ and one copy of $y^*$ fall into $S$. Denote the index by $(i,j)\in S$.

Using the same variance-based analysis from Theorem~\ref{thm:main1} and Theorem~\ref{thm:main2}, we have
\begin{align*}
    \E[C[i,j]] & ~ = \Omega(1) \cdot \sum_{k} 
 T'(X_{i,k}Y_{j,k}Z_{i,j})\\
    \mathrm{var}[C[i,j]] & ~ = \sum_{i',j'k,k'}T'(X_{i',k}Y_{j',k'}Z_{i,j})^2 \cdot \E[|X_i'|]\cdot \E[|Y_j'|].
\end{align*}

\begin{align*}
    & ~ \frac{\E[C[i,j]]}{\mathrm{var}[C[i,j]]^{1/2}}\\
    = & ~ \Omega(1) \cdot 
    \frac{\sum_{k} T'(X_{i,k}Y_{j,k}Z_{i,j})}
         {\effQ_{\tau}^2\cdot \sqrt{\sum_{i',j'k,k'} T'(X_{i',k}Y_{j',k'}Z_{i,j})^2 \cdot \Pi_{l=1}^{N} \partial_x(i'_l) \cdot \Pi_{l=N+1}^{2N} \partial_y(i'_l) \cdot \Pi_{l=1}^{N} \partial_y(j'_l) \cdot \Pi_{l=N+1}^{2N} \partial_x(j'_l) }}\\
    = & ~ \Omega(1)\cdot \frac{\effQ_{i,j}(Q_x^{\otimes N}\otimes Q_y^{\otimes N}, Q_y^{\otimes N}\otimes Q_x^{\otimes N}, T')}{\effQ_{\tau}^2}\\
    = & ~ \Omega(1),
\end{align*}
where the first step is by replacing $\E[|X_i'|],\E[|Y_j'|]$ from Lemma~\ref{lem:bound_x_i_y_y} and Eq.~\eqref{eq:g*}, the second step is by definition of $\eff_Q$ (Def.\ref{def:effQ_gamma}), the third step is because 
$\effQ_{\tau} = \effQ_{x,y}(Q_x^{\otimes N},Q_y^{\otimes N},T^{\otimes N})$ for any $(x,y)\in S_{\tau}$ from Eq.\ref{eq:StauDtaueffQtau} and the fact that $(i,j) \in S = S_{\tau}\otimes S_{\tau}^{\top}$.

As before, because the expectation exceeds the square root of variance, we can detect $x^*$ and $y^*$ by repeatedly running $T'$-matrix multiplication $O(\log n)$ times, which takes total running time $\tilde{O}(\rank(T'))$.

Overall, by repeating $(q^{2N})^{o(1)}$ times, we can boost the success probability to nearly $1$. Once we can detect if the planted pair ($x^*$, $y^*$) exists, we can do binary search to find them with comparably negligible time overhead.

The exponent $\omega_P$ we get is
\begin{align*}
    \omega_P = \frac{\log\left( (q^{2N})^{o(1)}\cdot \rank(T)^{2N} \right)}{\log n} = \frac{\log\left( (q^{2N})^{o(1)}\cdot \rank(T)^{2N} \right)}{\log \left( q^{2N}\gamma^NN^{-q^4} \right)} = \frac{\log \rank(T)}{\log (q\gamma^{1/2})} + o(1),
\end{align*}
where the second step is by \eqref{eq:q2n}.

\section{New Tensor Construction} \label{sec:newtensor}

In this section, we formally give the tensors summarized in Figure~\ref{fig:tensors} from the introduction. (In Figure~Figure~\ref{fig:tensors}, the bounds on $\omega_\ell$ from each of these tensors is calculated.) We begin with our new tensor $T_{2112}$.

For any $\eps>0$, define the rank-$5$ tensor $T_{2112}$ as the sum of the following five rank-1 tensors:

\begin{align*}&(\X_{0,0} + \X_{1,0} / \eps + \X_{0,1} / \eps^3 + \X_{1,1}) (\Y_{0,0} / \eps^3 + \Y_{1,0} + \Y_{0,1}  + \Y_{1,1} / \eps) (\eps^3 \Z_{0,0} + \eps^4 \Z_{1,0} + \eps^4 \Z_{0,1}  + \eps \Z_{1,1})/4 \\
+&(\X_{0,0} + \X_{1,0} / \eps - \X_{0,1} / \eps^3 - \X_{1,1}) (\Y_{0,0} / \eps^3 - \Y_{1,0} - \Y_{0,1}  + \Y_{1,1} / \eps) (\eps^3 \Z_{0,0} + \eps^4 \Z_{1,0} - \eps^4 \Z_{0,1}  - \eps \Z_{1,1})/4 \\
+&(\X_{0,0} - \X_{1,0} / \eps - \X_{0,1} / \eps^3 + \X_{1,1}) (\Y_{0,0} / \eps^3 - \Y_{1,0} + \Y_{0,1}  - \Y_{1,1} / \eps) (\eps^3 \Z_{0,0} - \eps^4 \Z_{1,0} + \eps^4 \Z_{0,1}  - \eps \Z_{1,1})/4 \\
+&(\X_{0,0} - \X_{1,0} / \eps + \X_{0,1} / \eps^3 - \X_{1,1}) (\Y_{0,0} / \eps^3 + \Y_{1,0} - \Y_{0,1}  - \Y_{1,1} / \eps) (\eps^3 \Z_{0,0} - \eps^4 \Z_{1,0} - \eps^4 \Z_{0,1}  + \eps \Z_{1,1})/4 \\
-& \X_{0,1} \Y_{0,0} \Z_{1,1} / \eps^5 \\
=& (\X_{0,0}\Y_{0,0} + \X_{0,1}\Y_{1,0} + \eps^3 \X_{1,1} \Y_{0,1} + \eps \X_{1,0} \Y_{1,1})\Z_{0,0}\\
+& (\eps^4 \X_{0,0}\Y_{0,1} + \X_{0,1}\Y_{1,1} + \eps \X_{1,1} \Y_{0,0} + \eps^3 \X_{1,0} \Y_{1,0})\Z_{0,1}\\
+& (\X_{1,0}\Y_{0,0} + \eps^4 \X_{1,1}\Y_{1,0} + \eps \X_{0,1} \Y_{0,1} + \eps^3 \X_{0,0} \Y_{1,1})\Z_{1,0}\\
+& (\X_{1,0}\Y_{0,1} + \X_{1,1}\Y_{1,1} + \eps \X_{0,0} \Y_{1,0})\Z_{1,1}.\end{align*}

(As discussed earlier, one might normally interpret this as a border rank expression, but here we substitute fixed values of $\eps>0$ and view it as a rank expression instead.) 
We can see that for any $\eps>0$, it has efficacies:
$$\eff_{0,0}(T) = \frac{1 + 1}{\sqrt{1 + 1 + \eps^6 + \eps^2}} = \sqrt{2} - O(\eps^2),$$
$$\eff_{0,1}(T) = \frac{\eps^4 + 1}{\sqrt{\eps^8 + 1 + \eps^2 + \eps^6}} = 1 - O(\eps^2),$$
$$\eff_{1,0}(T) = \frac{1 + \eps^4}{\sqrt{1 + \eps^8 + \eps^2 + \eps^6}} = 1 - O(\eps^2),$$
$$\eff_{1,1}(T) = \frac{1 + 1}{\sqrt{1 + 1 + \eps^2}} = \sqrt{2} - O(\eps^2).$$

Hence, $$\eff(T) = \sqrt{(\sqrt{2} - O(\eps^2))^2 + (1 - O(\eps^2))^2 + (1 - O(\eps^2))^2 + (\sqrt{2} - O(\eps^2))^2} = \sqrt{6} - O(\eps^2).$$

\subsection{Derivation of $T_{2112}$}

Although the rank expression above for $T_{2112}$ suffices for our algorithm, we give an alternate, fairly simple way to see why $T_{2112}$ has rank $5$; this is how we first found this tensor. 
Our rank expression for $T_{2112}$ was derived by a modification of the structural tensor $T_{(\mathbb{Z}/2)^2}$ of the group $(\mathbb{Z}/2)^2$ in the following way. $T_{(\mathbb{Z}/2)^2}$ is defined as

$$T_{(\mathbb{Z}/2)^2} = \sum_{a,b \in (\mathbb{Z}/2)^2} \X_{a} \Y_b \Z_{a+b},$$
and it has rank $4$ since it is the structural tensor of an Abelian group. We can expand its terms:

\begin{align*}X_{0,0} Y_{0,0} Z_{0,0}+X_{0,0} Y_{0,1} Z_{0,1}+X_{0,0} Y_{1,0} Z_{1,0}+X_{0,0} Y_{1,1} Z_{1,1}
\\+X_{0,1} Y_{0,0} Z_{0,1}+X_{0,1} Y_{0,1} Z_{0,0}+X_{0,1} Y_{1,0} Z_{1,1}+X_{0,1} Y_{1,1} Z_{1,0}
\\+X_{1,0} Y_{0,0} Z_{1,0}+X_{1,0} Y_{0,1} Z_{1,1}+X_{1,0} Y_{1,0} Z_{0,0}+X_{1,0} Y_{1,1} Z_{0,1}
\\+X_{1,1} Y_{0,0} Z_{1,1}+X_{1,1} Y_{0,1} Z_{1,0}+X_{1,1} Y_{1,0} Z_{0,1}+X_{1,1} Y_{1,1} Z_{0,0}\end{align*}

First, we rename some variables, swapping the names of $\X_{0,1} \leftrightarrow \X_{1,1}$ and the names of $\Y_{1,0} \leftrightarrow \Y_{1,1}$ to yield 

\begin{align*}X_{0,0} Y_{0,0} Z_{0,0}+X_{0,0} Y_{0,1} Z_{0,1}+X_{0,0} Y_{1,1} Z_{1,0}+X_{0,0} Y_{1,0} Z_{1,1}
\\+X_{1,1} Y_{0,0} Z_{0,1}+X_{1,1} Y_{0,1} Z_{0,0}+X_{1,1} Y_{1,1} Z_{1,1}+X_{1,1} Y_{1,0} Z_{1,0}
\\+X_{1,0} Y_{0,0} Z_{1,0}+X_{1,0} Y_{0,1} Z_{1,1}+X_{1,0} Y_{1,1} Z_{0,0}+X_{1,0} Y_{1,0} Z_{0,1}
\\+X_{0,1} Y_{0,0} Z_{1,1}+X_{0,1} Y_{0,1} Z_{1,0}+X_{0,1} Y_{1,1} Z_{0,1}+X_{0,1} Y_{1,0} Z_{0,0}\end{align*}

Since we just renamed variables, this tensor still has rank 4.

Next, we multiply some variables by powers of $\eps$. We multiply $X_{1,0}$ by $1/\eps$, multiply $X_{0,1}$ by $1/\eps^3$, multiply $Y_{1,1}$ by $1/\eps$, multiply $Y_{0,0}$ by $1/\eps^3$, multiply $Z_{0,0}$ by $\eps^3$, multiply $Z_{0,1}$ by $\eps^4$, multiply $Z_{1,0}$ by $\eps^4$, and multiply $Z_{1,1}$ by $\eps$, to yield

\begin{align*}X_{0,0} Y_{0,0} Z_{0,0}+ \eps^4 X_{0,0} Y_{0,1} Z_{0,1}+ \eps^3 X_{0,0} Y_{1,1} Z_{1,0}+ \eps X_{0,0} Y_{1,0} Z_{1,1}
\\+ \eps X_{1,1} Y_{0,0} Z_{0,1}+ \eps^3 X_{1,1} Y_{0,1} Z_{0,0}+X_{1,1} Y_{1,1} Z_{1,1}+ \eps^4 X_{1,1} Y_{1,0} Z_{1,0}
\\+X_{1,0} Y_{0,0} Z_{1,0}+X_{1,0} Y_{0,1} Z_{1,1}+ \eps X_{1,0} Y_{1,1} Z_{0,0}+ \eps^3 X_{1,0} Y_{1,0} Z_{0,1}
\\+\eps^{-5} X_{0,1} Y_{0,0} Z_{1,1}+ \eps X_{0,1} Y_{0,1} Z_{1,0}+X_{0,1} Y_{1,1} Z_{0,1}+X_{0,1} Y_{1,0} Z_{0,0}\end{align*}

Since we just multiplied variables by scalars, this did not change the rank, so this tensor still has rank $4$. (This is similar to an operation called a ``monomial degeneration'' or ``toric degeneration'' in the literature, although here we are thinking of $\eps$ as a fixed, small positive value rather than a formal variable.)

Finally, we delete the term $\eps^{-5} X_{0,1} Y_{0,0} Z_{1,1}$, yielding

\begin{align*}X_{0,0} Y_{0,0} Z_{0,0}+ \eps^4 X_{0,0} Y_{0,1} Z_{0,1}+ \eps^3 X_{0,0} Y_{1,1} Z_{1,0}+ \eps X_{0,0} Y_{1,0} Z_{1,1}
\\+ \eps X_{1,1} Y_{0,0} Z_{0,1}+ \eps^3 X_{1,1} Y_{0,1} Z_{0,0}+X_{1,1} Y_{1,1} Z_{1,1}+ \eps^4 X_{1,1} Y_{1,0} Z_{1,0}
\\+X_{1,0} Y_{0,0} Z_{1,0}+X_{1,0} Y_{0,1} Z_{1,1}+ \eps X_{1,0} Y_{1,1} Z_{0,0}+ \eps^3 X_{1,0} Y_{1,0} Z_{0,1}
\\+ \eps X_{0,1} Y_{0,1} Z_{1,0}+X_{0,1} Y_{1,1} Z_{0,1}+X_{0,1} Y_{1,0} Z_{0,0}\end{align*}

Since a single term has rank $1$, this is a rank-1 update to our tensor, so this new tensor has rank at most $5$. This is exactly our desired tensor $T_{2112}$.

We note that, since as $\eps \to 0$, $T_{2112}$ becomes 6 of the 8 terms of $\langle 2,2,2 \rangle$, one could add in the remaining two terms to give a relatively simple proof that the border rank of $\langle 2,2,2 \rangle$ is at most $5+2=7$. The fact that these 6 terms of $\langle 2,2,2 \rangle$ have border rank $6$ has also been independently observed by Vrana, although with a different border rank identity (and hence not yielding $T_{2112}$ specifically)~\cite{vrana}.

\subsection{Other Tensor Rank Bounds}

We give the other tensor rank bounds mentioned in the introduction.

Strassen~\cite{strassen} showed that $\langle 2,2,2 \rangle$ has rank at most $7$ via the following expression:
\begin{align*}&(\X_{1,1} + \X_{2,2})(\Y_{1,1} + \Y_{2,2})(\Z_{1,1} + \Z_{2,2}) \\+& (\X_{2,1} + \X_{2,2})(\Y_{1,1})(\Z_{2,1} - \Z_{2,2}) \\+& (\X_{1,1})(\Y_{1,2} - \Y_{2,2})(\Z_{1,2} + \Z_{2,2}) \\+& (\X_{2,2})(\Y_{2,1} - \Y_{1,1})(\Z_{1,1} + \Z_{2,1}) \\+& (\X_{1,1} + \X_{1,2})(\Y_{2,2})(-\Z_{1,1} + \Z_{1,2}) \\+& (\X_{2,1} - \X_{1,1})(\Y_{1,1} + \Y_{1,2})(\Z_{2,2}) \\+& (\X_{1,2} - \X_{2,2})(\Y_{2,1} + \Y_{2,2})(\Z_{1,1}) \\ =&(\X_{1,1} \Y_{1,1} + \X_{1,2}\Y_{2,1})\Z_{1,1} \\ +& (\X_{1,1}\Y_{1,2} + \X_{1,2}\Y_{2,2})\Z_{2,1} \\+& (\X_{2,1}\Y_{1,1} + \X_{2,2}\Y_{2,1})\Z_{1,2} \\ +& (\X_{2,1}\Y_{1,2} + \X_{2,2}\Y_{2,2})\Z_{2,2}\end{align*}

Winograd~\cite{winograd1971multiplication} showed via the Strassen-Winograd identity that the tensor $SW$, which consists of $7$ out of the $8$ terms of $\langle 2,2,2 \rangle$, has rank at most $6$ as follows:

\begin{align*}&(\X_{2,1} + \X_{2,2})(\Y_{2,1} + \Y_{2,2})(-\Z_{1,2} + \Z_{2,2}) \\+& (\X_{1,2})(\Y_{2,1})(\Z_{1,1} - \Z_{1,2} - \Z_{2,1} + \Z_{2,2}) \\+& (\X_{1,2} + \X_{2,2})(\Y_{1,2} - \Y_{2,2})(\Z_{2,1} - \Z_{2,2}) \\+& (\X_{1,2} + \X_{2,1} + \X_{2,2})(-\Y_{1,2} + \Y_{2,1} + \Y_{2,2})(\Z_{1,2} + \Z_{2,1} - \Z_{2,2}) \\+& (\X_{1,1} + \X_{1,2} + \X_{2,1} + \X_{2,2})(\Y_{1,2})(\Z_{1,2}) \\+& (\X_{2,1})(\Y_{1,1} + \Y_{1,2} - \Y_{2,1} - \Y_{2,2})(\Z_{2,1})  \\ =&( \X_{1,2}\Y_{2,1})\Z_{1,1} \\ +& (\X_{1,1}\Y_{1,2} + \X_{1,2}\Y_{2,2})\Z_{2,1} \\+& (\X_{2,1}\Y_{1,1} + \X_{2,2}\Y_{2,1})\Z_{1,2} \\ +& (\X_{2,1}\Y_{1,2} + \X_{2,2}\Y_{2,2})\Z_{2,2}\end{align*}

\section{Hashing gives an improvement for almost any tensor} \label{sec:hashingworks}

The goal of this section is to prove the following Theorem~\ref{thm:hashingintrogeneral_restate}.

\begin{theorem}[Restatement of Theorem~\ref{thm:hashingintrogeneral}]\label{thm:hashingintrogeneral_restate}
Suppose $T$ is a $\langle q,q,q_k \rangle$-sized tensor which consists of a subset of the terms of a matrix multiplication tensor, and the matrix $[(\eff_{i,j}(T))^2]_{i,j}$ has full rank. Let $\omega_\ell' := \frac{\log(\rank(T))}{\log(\eff(T))}$ be the exponent one would get from $T$ from applying Theorem~\ref{thm:mainintro}. Then, there is a non-decreasing, positive function $f_T : (0,1) \to \R_{>0}$ such that the bound of Theorem~\ref{thm:mainintro} can be improved to $$\omega_\ell \leq \omega_\ell' - f_T(\rho).$$
\end{theorem}

When $T$ is composed of a subset of the terms of a matrix multiplication tensor, $T(\X_{i',k}\Y_{j',k'}\Z_{i,j})\neq 0$ only if $i'=i$, $j'=j$ and $k'=k$. Thus, we can rewrite the $\effQ$ (Def.~\ref{def:effQ_gamma}) as 

\begin{align*}
\effQ_{i,j}(Q_x,Q_y,T) 
= & ~ \frac{1}{\partial_x(i)\partial_y(j)}\frac{\sum_{k \in [q_k]} T(\X_{i,k} \Y_{j,k} \Z_{i,j})}{\sqrt{\sum_{i', j' \in [q], k,k' \in [q_k]} T(\X_{i',k} \Y_{j',k'} \Z_{i,j})^2}}\\
= & ~ \frac{\eff_{i,j}(T)}{\partial_x(i)\partial_y(j)} 
\end{align*}

And $\gamma_{Q_x,Q_y}$ can be also rewritten as
\begin{align*}
\gamma_{Q_x,Q_y} 
= & ~ \prod_{i,j\in[q]}\left( \sum_{u,v\in [q]} Q_x[i,u] Q_y[j,v] (\effQ_{u,v}(Q_x,Q_y,T))^2 \right)^{P[i,j]}\\
= & ~ \prod_{i,j\in[q]}\left( \sum_{u,v\in [q]}\underbrace{\frac{Q_x[i,u]}{\partial_x(u)}}_{N_x[i,u]} \underbrace{\frac{Q_y[j,v]}{\partial_y(v)}}_{N_y[j,v]}\eff_{u,v}(T)^2\right)^{P[i,j]},
\end{align*}
where we define $N_x,N_y\in \R^{q\times q}$ as above. In other words, we normalize every column of $Q_x,Q_y$ to get $N_x,N_y$.

We begin with the key lemma behind our proof of Theorem~\ref{thm:hashingintrogeneral_restate}, which shows how we will pick the matrices $Q_x, Q_y$ for our hashing scheme.

\begin{lemma}\label{lem:constructqxqy}
Let $T$ be the tensor having the same property as that in Theorem~\ref{thm:hashingintrogeneral_restate}. There exist stochastic matrices $Q_x$, $Q_y$ such that $\gamma_{Q_x,Q_y} > \frac{1}{q^2}\sum_{i,j}\eff^2_i,j(T)$.
\end{lemma}
\begin{proof}
    Let $\eps > 0$ be a small constant. 
    Let $\avg = \frac{1}{q^2}\sum_{i,j}\eff^2_i,j(T)$. 
    Let $A = \eff^2(T) \in \R^{q\times q}$ be the matrix given by $A_{i,j} = \eff_{i,j}^2(T)$.

    Let $C \in \R^{q\times q}$ be the matrix defined by
    \begin{align*}
    C_{i,j} = \begin{cases}
    \avg + \eps, &\text{if $i=j=1$};\\
    \avg - \eps/(q^2-1), & \text{otherwise}.
    \end{cases}
    \end{align*}

    Suppose we can design $N_x,N_y$ such that 
    \[
    C = N_x\cdot A \cdot N_y^{\top},
    \]
    then $\gamma_{Q_x,Q_y} = \prod_{i,j\in [q]} C_{i,j}^{P[i,j]}$. According to Lemma~\ref{lem:beatavg}, $\gamma_{Q_x,Q_y} > \avg$ and therefore we conclude the lemma.
    
    In the following, we are going to prove that such $N_x$, $N_y$ exist in two steps. 
    For $j\in [q]$, let $c_j= \frac{1}{q}\sum_{i\in[q]} \eff^2_{i,j}(T)$ be the average of $j$-th column of $\eff^2(T)$.
    Define $B := \mathbf{1}_{q}^{\top} \cdot (c_1,\cdots,c_q) + \Delta \in \R^{q\times q}$. Here, $\mathbf{1}_{q}$ is an all-one vector of length $q$, and $\Delta$ is defined as follows, where $\delta:= \frac{q}{q+1}\eps$:
    \[
    \Delta_{i,j} := \begin{cases}
        0, & \text{ if $j \geq 2$};\\
        \delta, & \text{if $j=1$ and $i=1$}; \\
        - \delta/(q-1), & \text{if $j=1$ and $i\neq 1$}.
    \end{cases}
    \]
    The first step is to design $Q_x$ such that $N_x\cdot \eff^2(T) = B$, and the second step is to design $Q_y$ such that $B\cdot N_y^{\top} = C$.

    \paragraph{Step 1. Design $Q_x$}

        We first design $N_x$ as follows. 
        Let $N_x := \frac{1}{q}\cdot \mathbf{1}_{q\times q} + N_x'$, where $\mathbf{1}_{q\times q}$ is an all-one matrix of size $q\times q$. We will note that every entry in $N_x'$ is of order $O(\eps)$.

        By $N_x\cdot A = B$, we have
        \[
        (\frac{1}{q}\cdot \mathbf{1}_{q\times q} + N_x') \cdot A = \mathbf{1}_{q}^{\top} \cdot (c_1,\cdots,c_q) + \Delta,
        \]
        therefore,
        \begin{align*}
            N_x' & ~ = \Delta\cdot A^{-1}.\\
            (N_x')_{i,j} & ~ = 
            \begin{cases}
                \delta\cdot (A^{-1})_{1,j}, & \text{ if $i=1$ };\\
                -\frac{1}{q-1}\delta\cdot (A^{-1})_{1,j}, & \text{ if $i \geq 2$}.
            \end{cases}
        \end{align*}

        We design $Q_x$ as follows. For $j\in[q]$, let $z_j $ be variables. Let $(Q_x)_{i,j} := (N_x)_{i,j}\cdot z_j$. We will set $z_j$ so that every row of $Q_x$ sums up to $1$. Since all but the first row are all the same, we only need to care about the first row and the second row:
        \begin{align}\label{eq:twoeqaution}
        \begin{cases}
            &\sum_{j} (\frac{1}{q} + \delta\cdot (A^{-1})_{1,j}) \cdot z_j = 1;\\
            &\sum_{j} (\frac{1}{q} - \frac{\delta}{q-1}\cdot (A^{-1})_{1,j}) \cdot z_j = 1.
        \end{cases}
        \end{align}
We solve this pair of equations case by case.

\textbf{Case 1: When $A$ is a diagonal matrix.} In this special case, one solution to the Eq.~\ref{eq:twoeqaution} is
$z_1 =0$ and $z_2=z_3=\cdots=z_q = \frac{q}{q-1}$. In this case, all the entries in $Q_x$ are in $[0,1]$ and thus $Q_x$ is valid.

\textbf{Case 2: When $A$ is not diagonal.}
By Lemma~\ref{lem:part2}, there are two indices $j_1,j_2\in [q]$ that $(A^{-1})_{1,j_1} > 0$ and $(A^{-1})_{1,j_2} < 0$. Let $x = \delta\cdot (A^{-1})_{1,j_1} > 0$ and $y = \delta\cdot (A^{-1})_{1,j_2} < 0$. Then we let $z_{j_1} := \frac{q|y|}{|x|+|y|}$, $z_{j_2}:= \frac{q|x|}{|x|+|y|}$, and all $z_j:=0$ for all other $j$s. One can verify that this satisfies the Eq.~\eqref{eq:twoeqaution}.

        Since all entries of $(N_x)_{i,j} = 1/q + \Theta(\epsilon)$ and $0< z_{j_1},z_{j_2} < q$, for small enough $\eps$, we have every entry of $Q_x$ are in the range $(0,1)$, so that $Q_x$ is valid.

    \paragraph{Step 2. Design $Q_y$}
    We first show how to construct $N_y$ so that $B\cdot N_y^{\top} = C$.

    Let $b\in \R$ be a parameter. Define $N_y$ to be
    \begin{align*}
        (N_y^{\top})_{i,j} = \begin{cases}
        1, & \text{if $i=j=1$};\\
        0, & \text{if $i=1$, $j\geq 2$};\\
        b, & \text{if $i\geq 2$, $j=1$};\\
        \frac{1-b}{q-1}, & \text{if $i\geq 2$ and $j\geq 2$}.
        \end{cases}
    \end{align*}
    Under this design of $N_y$, the equation $B\cdot N_y^{\top} = C$ will become three small equations.
    \begin{equation}\label{eq:bNy=c}
        \begin{cases}
            & ~ c_1+\delta + b\cdot(c_2+\cdots+c_q) = \avg + \epsilon\\
            & ~ c_1 - \frac{\delta}{q-1} + b\cdot (c_2+\cdots+c_q) = \avg -\frac{\eps}{q^2-1}\\
            & ~ \frac{1-b}{q-1}(c_2+\cdots,c_q) =  \avg -\frac{\eps}{q^2-1}.
        \end{cases}
    \end{equation}

    Set $\delta = \frac{q}{q+1} \epsilon$, by solving Eq.~\ref{eq:bNy=c}, we get 
    \[
    b = \frac{1}{q} - \frac{c_1 (1 - 1/q) - \eps/(q+1)}{c_2 + \cdots + c_q} \leq \frac{1}{q}.
    \]

    Using the same method as in Step 1, we let 
    $(Q_x)_{i,j}:= (N_x)_{i,j}\cdot z_j$.
    We need to find $z_1,\cdots,z_q$ so that every row of $Q_x$ sums up to $1$. By setting $z_2=z_3=\cdots=z_q$, this reduces to two equations.
    \begin{equation*}
        \begin{cases}
        &~ z_1+(q-1)b\cdot z_2=1;\\
        &~ \frac{1-b}{q-1}\cdot z_2=1.
        \end{cases}
    \end{equation*}
    By solving these equations, we get 
    \[
    z_1 = 1 - \frac{b(q-1)}{1-b}, ~~~~ z_2=1/(1-b).
    \]
    Since $b\leq 1/q$, one can verify that every entry of $Q_x$ is in $[0,1]$, and $Q_x$ is a stochastic matrix.
\end{proof}

We now prove the helper lemmas for the above result:

\begin{lemma}\label{lem:beatavg}
Let $q$ be a positive integer, and suppose $a, \rho > 0$ and $1 \geq p > 1/q^2$. Then, for all sufficiently small $\eps>0$ we have
$$(a + \eps)^{p} \cdot \left(a - \frac{\eps}{q^2 - 1}\right)^{1 - p} > a.$$
\end{lemma}

\begin{proof}
Define $f(\eps) := (a + \eps)^{p} \cdot (a - \frac{\eps}{q^2 - 1})^{1 - p}$. Since $f(0) = a$, it suffices to prove that $f'(0) > 0$. Let $m = q^2 - 1$ so that $f(\eps) := (a + \eps)^{p} \cdot (a - \frac{\eps}{m})^{1 - p}$. Note that $(m+1) \cdot p > 1$ by definition of $p$. We have:

\begin{align*}
f'(\eps) &= p \cdot (a + \eps)^{p-1} \cdot (a - \frac{\eps}{m})^{1 - p} + (a + \eps)^{p} \cdot \frac{p-1}{m}(a - \frac{\eps}{m})^{- p} \\&= \frac{(amp + ap - a - \eps) (a + \eps)^{p-1} (a - \eps/m)^{-p}}{m}.
\end{align*}
Hence, as desired,
\begin{align*}
f'(0) &=  \frac{(amp + ap - a) (a)^{p-1} (a )^{-p}}{m} \\  &=  \frac{mp + p - 1}{m} > 0.
\end{align*}
\end{proof}

\begin{lemma} \label{lem:part1}
Suppose $A \in \R_{\geq 0}^{q \times q}$ for $q \geq 2$ is a full-rank matrix with nonnegative entries such that at least one of its rows has at least two nonzero entries. Then, one can permute the columns of $A$ so that it has the following property: For every row of $A$, if its first entry is nonzero, then another one of its entries is also nonzero.
\end{lemma}

\begin{proof} 
It suffices to prove that there is a column $j$ of $A$ such that: for every $i$ with $A[i,j] \neq 0$, there exists an $i' \neq i$ with $A[i',j] \neq 0$. We can then permute the columns of $A$ so that $j$ becomes the first column as desired.

Assume to the contrary that there were no such $j$. Since $A$ has full rank, we know every column of $A$ has a nonzero entry. It follows that for every $j$, there is a row with a nonzero entry in column $j$ but no other column. Since $A$ has the same number of rows and columns, this means every row and every column of $A$ has exactly one nonzero entry. This contradicts our assumption that $A$ has a row with at least two nonzero entries.
\end{proof}

\begin{lemma}\label{lem:part2}
Suppose $A \in \R_{\geq 0}^{q \times q}$ for $q \geq 2$ is a full-rank matrix with nonnegative entries such that at least one of its rows has at least two nonzero entries. Then, one can permute the columns of $A$ so that it has the following property: In the top row of the matrix $A^{-1}$, there is at least one positive entry and at least one negative entry.
\end{lemma}

\begin{proof}
Applying Lemma~\ref{lem:part1}, we may assume that for every row $i$ of $A$, if $A[i,1] \neq 0$, then there is an $j \neq 1$ such that $A[i,j] \neq 0$.

We claim first that the top row of $A^{-1}$ must have at least two nonzero entries. Assume to the contrary that this is not the case. It must have at least one nonzero entry since $A^{-1}$ has full rank, so it has exactly one nonzero entry. Suppose it is in column $j$, so $A^{-1}[1,j] \neq 0$ and $A^{-1}[1,j'] = 0$ for all $j' \neq j$. We know that the top-right entry of the product $A^{-1} A$ is $1$, so it follows that $A[j,1] = 1/A^{-1}[1,j] \neq 0$. By the property of the previous paragraph, there is a $i \neq 1$ such that $A[j,i] \neq 0$. It follows that entry $(1,i)$ of the product $A^{-1} A$ is equal to $\sum_{k} A^{-1}[1,k] \cdot A[k,i] = A^{-1}[1,j] \cdot A[j,i] \neq 0$, contradicting the fact that $A^{-1} A$ is the identity matrix whose $(1,i)$ entry is $0$. This proves the claim.

Now, we know the top row of $A^{-1}$ has at least two nonzero entries. We claim that the nonzero entries of the top row of $A^{-1}$ cannot all be positive or all be negative, which will complete the proof. Assume to the contrary that they are all positive (the all negative case is identical), and as before, suppose $A^{-1}[1,j] > 0$ is one of the nonzero entries of the first row of $A^{-1}$. Since the first row has at least two nonzero entries, we may assume $j \neq 1$. Since $A$ has full rank, there is an $i$ such that $A[j,i] \neq 0$, and since $A$ has nonnegative entries, we further have $A[j,i] > 0$ and $A[j',i] \geq 0$ for all $j'$. It follows that entry $(1,j)$ of the product $A^{-1} A$ is $\sum_k A^{-1}[1,k] A[k,i] \geq A^{-1}[1,j] A[j,i] > 0$, contradicting again that it must equal $0$. This completes the proof.
\end{proof}

Finally we conclude the main proof:

\begin{proof}[Proof of Theorem~\ref{thm:hashingintrogeneral_restate}]
By Lemma~\ref{lem:constructqxqy}, we can construct stochastic matrices $Q_x$, $Q_Y$ such that $\gamma_{Q_x,Q_y} > \frac{1}{q^2}\sum_{i,j} \eff_{i,j}^2(T)$. By Theorem~\ref{thm:main_new_2}, we have
\[
\omega_{P} \leq \frac{\log \rank(T)}{\log(\Peff(T))},
\]
where $\Peff(T) \geq q \gamma_{Q_x,Q_y}^{1/2} > \eff(T).$ It's clear from the proof that the difference between $\Peff(T)$ and $\eff(T)$ is a function of $\rho$. The function $f_T(\rho)$ is an non-decreasing function since we can always reduce larger $\rho$ to smaller $\rho$.
\end{proof}

\appendix
\section*{Appendix}

\section{Aggregation time} \label{sec:agg}

In the algorithms throughout this paper, we assumed that the input vectors $x_1,\ldots,x_n$,$y_1,\ldots,y_n$ have long enough length $d$ which is polynomial in $n$, i.e., $d = q^N$ ($q$ and $N$ are defined in Algorithm~\ref{alg:intro}), whereas we would like our algorithm to work for the information-theoretically minimum $d = O(\log n/\rho^2)$ (recall that $\rho\in(0,1)$ is the correlation of the planted pair). Furthermore, we assumed that the aggregation step of the algorithm (lines~\ref{line:construct_A},~\ref{line:construct_B} of Algorithm~\ref{alg:intro}) takes negligible time compared to the rest of the algorithm. (See footnote \ref{footnote:aggregation} above.) However, if implemented naively, the aggregation step can actually take time $q^N\cdot d$, which can potentially be the slowest step of the algorithm. In this section, we show how the ``compressed matrix'' technique of~\cite{karppa2018faster} can be used to require only the smaller $d = O(\log n/\rho^2)$, and simultaneously decrease the aggregation time. 

Suppose the given vectors $x \in \{-1,1\}^d$ have short length $d = O(\log n/\rho^2)$ and we want a long enough vector $x' \in \{-1,1\}^m$ to be used in our algorithm, for some $m = \poly(n)$. To ``prolong'' the vector, we first pick $r \leq d$ such that ${d \choose r} = m$, and define for every subset $S \subseteq [d]$ with $|S| = r$, the entry
\begin{align}\label{eq:compressed}
x'_S = \prod_{j\in S} x_j.
\end{align}
This new vector $x'$ will be \emph{implicitly} used as the true input vector in our algorithm; in fact, we will never compute $x'$, but rather the aggregation of all such ``prolonged'' vectors defined as follows.

\begin{definition}[Aggregation problem]\label{def:aggregation_problem}
Given $x_1,\ldots,x_g \in \{-1,1\}^d$ and $r \leq d$. Let $m = {d\choose r}$ and $x_1',\ldots,x_g'$ be defined as in Eq.~\eqref{eq:compressed}.
The goal is to compute, for every $j\in[m]$, the value $X_j = \sum_{i\in [g]} x_i'[j].$ (I.e., the goal is to compute the vector $\sum_{i\in [g]} x_i'$.)
\end{definition}

Note that the aggregation vectors $a_i,b_i$ in line~\ref{line:construct_a_b} of Algorithm~\ref{alg:intro} can be computed by solving this aggregation problem $O(q^N)$ times, and hence we construct our matrices $A$ and $B$ from lines~\ref{line:construct_A} and \ref{line:construct_B}.

\begin{lemma}[\cite{karppa2018faster,alman2018illuminating}]\label{lem:quick_agg}
The Aggregation problem defined above (Def.~\ref{def:aggregation_problem}) can be solved in $\textbf{MM}(m^{1/2+o(1)}, g, m^{1/2+o(1)})$ time, where $\textbf{MM}(a,b,c)$ is the time to multiply a matrix of size $a\times b$ with another matrix of size $b\times c$.
\end{lemma}
\begin{proof}
The proof is identical to the aggregation algorithm used by prior light bulb algorithms, such as \cite[page 6, second and third paragraphs]{alman2018illuminating}.
\end{proof}

\begin{lemma}\label{lem:aggre_time}
While running Algorithm~\ref{alg:intro} with input vectors $x_1,\ldots,x_n$, $y_1,\ldots,y_n \in \{-1,1\}^d$, and $\langle q,q,q \rangle$-sized tensor $T$,
the matrix $A$ and $B$ defined in line~\ref{line:construct_A} and line~\ref{line:construct_B} can be computed in $n^{\frac{1+\omega/2}{\log_q \eff(T)} + o(1)}$ time.
\end{lemma}
\begin{proof}
Since $A$ and $B$ are constructed in the same way, we only analyze $A$.

Let $g$ be such that $g^2\cdot |\{i,j\in[q^N]: \eff_{i,j}(T^{\otimes N})\geq g^2\}|$ is maximized, as defined in line~\ref{line:size_is_g} of the algorithm. Note that $\eff_{i,j}(T)$ cannot exceed $q$ by the Cauchy-Schwarz inequality (it is maximized when $T=\langle q,q,q \rangle$), and so $g\leq \sqrt{\max_{i,j}\eff_{i,j}(T^{\otimes N})} \leq \sqrt{q^N}$. 

In line~\ref{line:construct_a_b}, each $a_i$ aggregates together $|X_i|$ vectors. We have $\E[|X_i|] = nt/q^N = g$ since we set $t = q^Ng/n$. By a Chernoff bound, $|X_i| = O(g)$  for each $i$ with high probability. 

Let $m=q^N$. Since the tensor $T$ has size $q_k = q$, the desired length of input vectors $d$ is also $m$. By Lemma~\ref{lem:quick_agg}, calculating all the $a_i$ can be done in time
\[
m \cdot \textbf{MM}(\sqrt{d},g,\sqrt{d}) = m\cdot \textbf{MM}(\sqrt{m},g,\sqrt{m}) \leq m\cdot g^{\omega}\cdot (\sqrt{m}/g)^2 = m^2g^{\omega-2} \leq m^{2+\frac{\omega-2}{2}}.
\]

(Here we used that, since $\sqrt{m} \geq g$, we have $\textbf{MM}(\sqrt{m},g,\sqrt{m}) \leq (\sqrt{m}/g)^2 \cdot \textbf{MM}(g,g,g)$.)

Since $\rank(T)^N = n^{\frac{\log \rank(T)}{\log \eff(T)}}$, we have $N = \log n /\log \eff(T)$.
Thus, this is the desired running time since $m=q^N = n^{\log q/\log \eff(T)}$.
\end{proof}
\begin{remark}
When $T$ is a matrix multiplication tensor, $\log_q \eff(T) = 1.5$, so the aggregation time is $n^{\frac{1+\omega/2}{\log_q \eff(T)} + o(1)} = O(n^{\frac{2+\omega}{3}}) < O(n^{\frac{2\omega}{3}})$. The aggregation time exponent $\frac{2+\omega}{3}$ is less than $\frac{2\omega}{3}$, so aggregation takes negligible time compared to the remainder of the algorithm.
\end{remark}

\begin{lemma}
If there is a $\langle q,q,q \rangle$-sized tensor $T$ with 
\[
\frac{\log \rank(T)}{\log \eff(T)} < \frac{2\omega}{3},
\]
then there is another tensor $T'$ that can solve light bulb problem in time $n^{\frac{2\omega}{3} - \eps}$ for some $\eps > 0$.
\end{lemma}
\begin{proof}
Let $N$ be a large enough constant and we let $T' = T^{\otimes \delta N}\otimes \langle q,q,q\rangle^{\otimes (1-\delta)N}$ for some $\delta \in(0,1)$ to be determined. So 
\[
\log \rank(T') = N\cdot(\delta \log \rank(T) + (1-\delta)\log \rank(\langle q,q,q\rangle))
\]
and 
\[
\log \eff(T') = N\cdot(\delta \log \eff(T) + (1-\delta)\log \eff(\langle q,q,q\rangle)).
\]
Since $\frac{\log \rank(\langle q,q,q\rangle)}{\log \eff(\langle q,q,q\rangle)} = \frac{2\omega}{3}$, choosing any $\delta > 0$ results in $\frac{\log \rank(T')}{\log \eff(T')} < \frac{2\omega}{3}$.

By Lemma~\ref{lem:aggre_time}, the aggregation time of $T'$ is $n^{\frac{1+\omega/2}{\log_q \eff(T')} + o(1)} \leq n^{\frac{1+\omega/2}{(1-\delta)1.5} + o(1)}$. We can choose a small enough $\delta$ so that $\frac{1}{1-\delta}\cdot \frac{1+\omega/2}{1.5} < \frac{2\omega}{3}.$ Thus, the running time for both the main procedure and the aggregation part while using tensor $T'$ is small.
\end{proof}

\bibliographystyle{alpha}
\bibliography{ref}

\newcommand{\etalchar}[1]{$^{#1}$}
\begin{thebibliography}{HSHVDG16}

\bibitem[ACW16]{alman2016polynomial}
Josh Alman, Timothy~M Chan, and Ryan Williams.
\newblock Polynomial representations of threshold functions and algorithmic
  applications.
\newblock In {\em 2016 IEEE 57th Annual Symposium on Foundations of Computer
  Science (FOCS)}, pages 467--476. IEEE, 2016.

\bibitem[ACW20]{alman2020faster}
Josh Alman, Timothy~M Chan, and Ryan Williams.
\newblock Faster deterministic and las vegas algorithms for offline approximate
  nearest neighbors in high dimensions.
\newblock In {\em Proceedings of the Fourteenth Annual ACM-SIAM Symposium on
  Discrete Algorithms}, pages 637--649. SIAM, 2020.

\bibitem[AIR18]{andoni2018approximate}
Alexandr Andoni, Piotr Indyk, and Ilya Razenshteyn.
\newblock Approximate nearest neighbor search in high dimensions.
\newblock In {\em Proceedings of the International Congress of Mathematicians:
  Rio de Janeiro 2018}, pages 3287--3318. World Scientific, 2018.

\bibitem[Alm18]{alman2018illuminating}
Josh Alman.
\newblock An illuminating algorithm for the light bulb problem.
\newblock In {\em 2nd Symposium on Simplicity in Algorithms (SOSA 2019)}.
  Schloss Dagstuhl-Leibniz-Zentrum fuer Informatik, 2018.

\bibitem[ALRW17]{andoni2017optimal}
Alexandr Andoni, Thijs Laarhoven, Ilya Razenshteyn, and Erik Waingarten.
\newblock Optimal hashing-based time-space trade-offs for approximate near
  neighbors.
\newblock In {\em Proceedings of the Twenty-Eighth Annual ACM-SIAM Symposium on
  Discrete Algorithms}, pages 47--66. SIAM, 2017.

\bibitem[AR15]{andoni2015optimal}
Alexandr Andoni and Ilya Razenshteyn.
\newblock Optimal data-dependent hashing for approximate near neighbors.
\newblock In {\em Proceedings of the forty-seventh annual ACM symposium on
  Theory of computing}, pages 793--801, 2015.

\bibitem[AW15]{alman2015probabilistic}
Josh Alman and Ryan Williams.
\newblock Probabilistic polynomials and hamming nearest neighbors.
\newblock In {\em 2015 IEEE 56th Annual Symposium on Foundations of Computer
  Science (FOCS)}, pages 136--150. IEEE, 2015.

\bibitem[AW21]{alman2021refined}
Josh Alman and Virginia~Vassilevska Williams.
\newblock A refined laser method and faster matrix multiplication.
\newblock In {\em Proceedings of the 2021 ACM-SIAM Symposium on Discrete
  Algorithms (SODA)}, pages 522--539. SIAM, 2021.

\bibitem[Bin80a]{bini1980border}
Dario Bini.
\newblock Border rank of ap$\times$ q$\times$ 2 tensor and the optimal
  approximation of a pair of bilinear forms.
\newblock In {\em International Colloquium on Automata, Languages, and
  Programming}, pages 98--108. Springer, 1980.

\bibitem[Bin80b]{bini1980relations}
Dario Bini.
\newblock Relations between exact and approximate bilinear algorithms.
  applications.
\newblock {\em Calcolo}, 17(1):87--97, 1980.

\bibitem[BL16]{blaser2016degeneration}
Markus Bl{\"a}ser and Vladimir Lysikov.
\newblock On degeneration of tensors and algebras.
\newblock {\em arXiv preprint arXiv:1606.04253}, 2016.

\bibitem[CGLV19]{conner2019tensors}
Austin Conner, Fulvio Gesmundo, Joseph~M Landsberg, and Emanuele Ventura.
\newblock Tensors with maximal symmetries.
\newblock {\em arXiv preprint arXiv:1909.09518}, 2019.

\bibitem[Cha02]{rounding}
Moses~S Charikar.
\newblock Similarity estimation techniques from rounding algorithms.
\newblock In {\em STOC}, 2002.

\bibitem[CHL22]{conner2022bad}
Austin Conner, Hang Huang, and JM~Landsberg.
\newblock Bad and good news for strassen’s laser method: Border rank of perm
  3 and strict submultiplicativity.
\newblock {\em Foundations of Computational Mathematics}, pages 1--39, 2022.

\bibitem[CU13]{cohn2013fast}
Henry Cohn and Christopher Umans.
\newblock Fast matrix multiplication using coherent configurations.
\newblock In {\em Proceedings of the twenty-fourth annual ACM-SIAM symposium on
  Discrete algorithms}, pages 1074--1086. Society for Industrial and Applied
  Mathematics, 2013.

\bibitem[CV22]{vrana}
Matthias Christandl and Péter Vrana.
\newblock personal communication, 2022.

\bibitem[CW82]{CoppersmithW82}
Don Coppersmith and Shmuel Winograd.
\newblock On the asymptotic complexity of matrix multiplication.
\newblock {\em {SIAM} J. Comput.}, 11(3):472--492, 1982.

\bibitem[DS13]{stothers}
A.M. Davie and A.~J. Stothers.
\newblock Improved bound for complexity of matrix multiplication.
\newblock {\em Proceedings of the Royal Society of Edinburgh, Section: A
  Mathematics}, 143:351--369, 4 2013.

\bibitem[Dub10]{dubiner2010bucketing}
Moshe Dubiner.
\newblock Bucketing coding and information theory for the statistical
  high-dimensional nearest-neighbor problem.
\newblock {\em IEEE Transactions on Information Theory}, 56(8):4166--4179,
  2010.

\bibitem[FBH{\etalchar{+}}22]{fawzi2022discovering}
Alhussein Fawzi, Matej Balog, Aja Huang, Thomas Hubert, Bernardino
  Romera-Paredes, Mohammadamin Barekatain, Alexander Novikov, Francisco~J
  R~Ruiz, Julian Schrittwieser, Grzegorz Swirszcz, et~al.
\newblock Discovering faster matrix multiplication algorithms with
  reinforcement learning.
\newblock {\em Nature}, 610(7930):47--53, 2022.

\bibitem[Har21]{harris2021improved}
David~G Harris.
\newblock Improved algorithms for boolean matrix multiplication via
  opportunistic matrix multiplication.
\newblock {\em arXiv preprint arXiv:2109.13335}, 2021.

\bibitem[HJMS22]{homs2022bounds}
Roser Homs, Joachim Jelisiejew, Mateusz Micha{\l}ek, and Tim Seynnaeve.
\newblock Bounds on complexity of matrix multiplication away from
  coppersmith--winograd tensors.
\newblock {\em Journal of Pure and Applied Algebra}, 226(12):107142, 2022.

\bibitem[HSHVDG16]{huang2016strassen}
Jianyu Huang, Tyler~M Smith, Greg~M Henry, and Robert~A Van De~Geijn.
\newblock Strassen's algorithm reloaded.
\newblock In {\em SC'16: Proceedings of the International Conference for High
  Performance Computing, Networking, Storage and Analysis}, pages 690--701.
  IEEE, 2016.

\bibitem[KK19]{karppa2019probabilistic}
Matti Karppa and Petteri Kaski.
\newblock Probabilistic tensors and opportunistic boolean matrix
  multiplication.
\newblock In {\em Proceedings of the Thirtieth Annual ACM-SIAM Symposium on
  Discrete Algorithms}, pages 496--515. SIAM, 2019.

\bibitem[KKK18]{karppa2018faster}
Matti Karppa, Petteri Kaski, and Jukka Kohonen.
\newblock A faster subquadratic algorithm for finding outlier correlations.
\newblock {\em ACM Transactions on Algorithms (TALG)}, 14(3):1--26, 2018.

\bibitem[LG14]{legall}
Fran{\c{c}}ois Le~Gall.
\newblock Powers of tensors and fast matrix multiplication.
\newblock In {\em ISSAC}, pages 296--303, 2014.

\bibitem[Pan78]{pan1978strassen}
V~Ya Pan.
\newblock Strassen's algorithm is not optimal trilinear technique of
  aggregating, uniting and canceling for constructing fast algorithms for
  matrix operations.
\newblock In {\em 19th Annual Symposium on Foundations of Computer Science
  (sfcs 1978)}, pages 166--176. IEEE, 1978.

\bibitem[Pan18]{pan2018fast}
Victor~Y Pan.
\newblock Fast feasible and unfeasible matrix multiplication.
\newblock {\em arXiv preprint arXiv:1804.04102}, 2018.

\bibitem[Sch81]{schonhage1981partial}
Arnold Sch{\"o}nhage.
\newblock Partial and total matrix multiplication.
\newblock {\em SIAM Journal on Computing}, 10(3):434--455, 1981.

\bibitem[Str69]{strassen}
Volker Strassen.
\newblock Gaussian elimination is not optimal.
\newblock {\em Numerische mathematik}, 13(4):354--356, 1969.

\bibitem[Str73]{strassen1973vermeidung}
Volker Strassen.
\newblock Vermeidung von divisionen.
\newblock {\em Journal f{\"u}r die reine und angewandte Mathematik},
  264:184--202, 1973.

\bibitem[Str87]{strassenlaser1}
V.~Strassen.
\newblock Relative bilinear complexity and matrix multiplication.
\newblock {\em J. reine angew. Math. (Crelles Journal)}, 375--376:406--443,
  1987.

\bibitem[Val88]{lightbulb}
Leslie~G Valiant.
\newblock Functionality in neural nets.
\newblock In {\em AAAI}, 1988.

\bibitem[Val12]{valiant2012finding}
Gregory Valiant.
\newblock Finding correlations in subquadratic time, with applications to
  learning parities and juntas.
\newblock In {\em 2012 IEEE 53rd Annual Symposium on Foundations of Computer
  Science}, pages 11--20. IEEE, 2012.

\bibitem[Wil12]{v12}
Virginia~Vassilevska Williams.
\newblock Multiplying matrices faster than {C}oppersmith-{W}inograd.
\newblock In {\em STOC}, pages 887--898, 2012.

\bibitem[Win71]{winograd1971multiplication}
Shmuel Winograd.
\newblock On multiplication of 2$\times$ 2 matrices.
\newblock {\em Linear algebra and its applications}, 4(4):381--388, 1971.

\end{thebibliography}

\end{document}